\documentclass[11pt]{article}

\usepackage{fullpage}
\usepackage{amsmath,amsfonts,amsthm,mathrsfs,mathpazo,xspace,hyperref,graphicx}
\usepackage{endnotes}
\usepackage{color}
\usepackage{bm}
\usepackage{times}
\usepackage{amssymb,latexsym}
\usepackage{enumitem}

\newtheorem{theorem}{Theorem}
\newtheorem{proposition}[theorem]{Proposition}
\newtheorem{lemma}[theorem]{Lemma}
\newtheorem{claim}[theorem]{Claim}

\newtheorem{corollary}[theorem]{Corollary}
\newtheorem{definition}[theorem]{Definition}
\newtheorem{observation}[theorem]{Observation}
\newtheorem{conjecture}[theorem]{Conjecture}

\newcommand{\beq}{\begin{eqnarray}}
\newcommand{\eeq}{\end{eqnarray}}

\newcommand{\ket}[1]{|#1\rangle}
\newcommand{\bra}[1]{\langle#1|}
\newcommand{\Tr}{\mbox{\rm Tr}}
\newcommand{\Id}{\ensuremath{\mathop{\rm Id}\nolimits}}
\newcommand{\Es}[1]{\textsc{E}_{#1}}

\newcommand{\N}{\ensuremath{\mathbb{N}}}

\newcommand{\eps}{\varepsilon}

\newcommand{\MA}{\mathrm{MA}}

\newcommand{\MIP}{\mathrm{MIP}}
\newcommand{\AMIP}{\mathrm{MIP}_{ac}}

\newcommand{\NEXP}{\mathrm{NEXP}}
\newcommand{\EXP}{\mathrm{EXP}}
\newcommand{\TIME}{\mathrm{TIME}}
\newcommand{\DTIME}{\mathrm{DTIME}}

\newcommand{\classfont}{\mathrm}

\newcommand{\MIPstar}{\classfont{MIP}^*}

\DeclareMathOperator{\poly}{poly}
\DeclareMathOperator{\expt}{exp}

\newcommand{\setft}[1]{\mathrm{#1}}

\newcommand{\density}[1]{\setft{D}\left(#1\right)}

\newcommand{\pos}[1]{\setft{Pos}\left(#1\right)}

\newcommand{\Ima}{\ensuremath{\textrm{Im}}}

\newcommand{\Pileq}[1]{\Pi_{\leq #1}}

\newcommand{\tildPiPxa}{\tilde{P}}
\newcommand{\tildPiQyb}{\tilde{Q}}

\newcommand{\pili}{\Pi_{ \leq i}}
\newcommand{\pilj}{\Pi_{ \leq j}}

\newcommand{\pieip}{\Pi_{=i+1}}
\newcommand{\piejp}{\Pi_{=j+1}}

\newcommand{\PiPxa}{P}
\newcommand{\PiQyb}{Q}

\begin{document}

\title{Interactive proofs with approximately commuting provers}
\author{Matthew Coudron\thanks{Massachusetts Institute of Technology. Email 
\texttt{\href{mailto:mcoudron@mit.edu}{\color{black}mcoudron@mit.edu}}. Part of this work was accomplished while the author was a visitor at the Institute for Quantum Information and Matter (IQIM) at the California Institute of Technology.}
\qquad Thomas Vidick\thanks{Department of Computing and Mathematical Sciences, California Institute of Technology, Pasadena CA 91106, USA. Email 
\texttt{\href{mailto:vidick@cms.caltech.edu}{\color{black}vidick@cms.caltech.edu}}}}

\maketitle

\begin{abstract}
The class $\MIP^*$ of promise problems that can be decided through an interactive proof system with multiple entangled provers provides a complexity-theoretic framework for the exploration of the nonlocal properties of entanglement. Very little is known in terms of the power of this class. 
The only proposed approach for establishing upper bounds is based on a hierarchy of semidefinite programs introduced independently by Pironio et al. and Doherty et al. in 2006. This hierarchy converges to a value, the field-theoretic value, that is only known to coincide with the provers' maximum success probability in a given proof system under a plausible but difficult mathematical conjecture, Connes' embedding conjecture.  No bounds on the rate of convergence are known.

\smallskip

We introduce a rounding scheme for the hierarchy, establishing that any solution to its $N$-th level can be mapped to a strategy for the provers in which measurement operators associated with distinct provers have pairwise commutator bounded by $O(\ell^2/\sqrt{N})$ in operator norm, where $\ell$ is the number of possible answers per prover. 

\smallskip

Our rounding scheme motivates the introduction of a variant of quantum multiprover interactive proof systems, called $\MIP_\delta^*$, in which the soundness property is required to hold against provers allowed to operate on the {same} Hilbert space as long as the commutator of operations performed by distinct provers has norm at most $\delta$.  Our rounding scheme implies the upper bound $\MIP_\delta^* \subseteq \DTIME(\exp(\exp(\poly)/\delta^2))$. In terms of lower bounds we establish that $\MIP^*_{2^{-\poly}}$, with completeness $1$ and soundness $1-2^{-\poly}$, contains $\NEXP$.  The relationship of $\MIP_\delta^*$ to $\MIPstar$ has connections with the mathematical literature on approximate commutation. Our rounding scheme gives an elementary proof that the Strong Kirchberg Conjecture implies that $\MIPstar$ is computable. We also discuss applications to device-independent cryptography.

\end{abstract}

\section{Introduction}

In a multiprover interactive proof system, a \emph{verifier} with bounded resources (a polynomial-time Turing machine) interacts with multiple all-powerful but non-communicating \emph{provers} in an attempt to verify the truth of a mathematical statement --- the membership of some input $x$, a string of bits, in a language $L$, such as $3$-SAT. The provers always collaborate to maximize their chances of making the verifier accept the statement, and their maximum probability of success in doing so is called the \emph{value} $\omega = \omega(x)$ of the protocol. (We will sometimes refer to a given protocol as an ``interactive game''.) A proof system's \emph{completeness} $c$ is the smallest value of $\omega(x)$ over all $x\in L$, while its soundness $s$ is the largest value of $\omega(x)$ over $x\notin L$; a protocol is sound if $s<c$.  

The class of all languages that have multiprover interactive proof systems with $c\geq 2/3$ and $s\leq 1/3$, denoted $\MIP$, is a significant broadening of its non-interactive, single-prover analogue $\MA$, as is witnessed by the characterization $\MIP = \NEXP$~\cite{BabForLun91CC}. This result is one of the cornerstones on which the PCP theorem~\cite{AroSaf98JACM,AroLunMotSudSze98JACM} was built, with consequences ranging from cryptography~\cite{BenGolKilWig88STOC} to hardness of approximation~\cite{FeiGolLovSafSze96JACM}. 

%IP= PSPACE \cite{LunForKarNis92JACM,Sha92}
Quantum information suggests a natural extension of the class $\MIP$.  The laws of quantum mechanics assert that, in the physical world, a set of non-communicating provers may share an arbitrary entangled quantum state, a physical resource which strictly extends their set of strategies but provably does not allow them to communicate.  The corresponding extension of $\MIP$ is the class $\MIP^*$ of all languages that have multiprover interactive proof systems with entangled provers~\cite{KobMat03JCSS}. 

Physical intuition for the significance of the prover's new resource, entanglement, dates back to Einstein, Podolsky and Rosen's paradoxical account~\cite{epr} of the consequences of quantum entanglement, later clarified through Bell's pioneering work~\cite{Bell:64a}.
To state the relevance of Bell's results more precisely in our context we first introduce the mathematical formalism used by Bell to model locality. With each prover's private space is associated a separate Hilbert space. The joint quantum state of the provers is specified by a unit vector $\ket{\Psi}$ in the tensor product of their respective Hilbert spaces. Upon receiving its query from the verifier, each prover applies a local measurement (a positive operator supported on its own Hilbert space) the outcome of which is sent back to the verifier as its answer. The supremum of the provers' probability of being accepted by the verifier, taken over all Hilbert spaces, states in their joint tensor product, and local measurements, is called the entangled value $\omega^*$ of the game.  The analogue quantity for ``classical" provers (corresponding to shared states which are product states) is denoted $\omega$.   

Bell's work and the extensive literature on Bell inequalities~\cite{Clauser:69a,Arvind:02} and quantum games~\cite{CHTW04} establishes that there are protocols, or interactive games, for which $\omega^* > \omega$. This simple fact has important consequences for interactive proof systems. First, a proof system sound with classical provers may no longer be so in the presence of entanglement. Cleve et al.~\cite{CHTW04} exhibit a class of restricted interactive proof systems, XOR proof systems, such that the class with classical provers equals $\NEXP$ while the same proof systems with entangled provers cannot decide any language beyond $\EXP$. Second, the completeness property of a proof system may also increase through the provers' use of entanglement. As a result optimal strategies may require the use of arbitrarily large Hilbert spaces for the provers --- no explicit bound on the dimension of these spaces is known as a function of the size of the game. In fact no better upper bound on the class $\MIPstar$ is known other than its languages being recursively enumerable: they may not even be decidable! This unfortunate state of affairs stems from the fact that, while the value $\omega^*$ may be approached from below through exhaustive search in increasing dimensions, there is no verifiable criterion for the termination of such a procedure. 

\paragraph{Bounding entangled-prover strategies.}
The question of deriving algorithmic methods for placing upper bounds on the entangled value $\omega^*$ of a given protocol has long frustrated researchers' efforts. Major progress came in 2006 through the introduction of a hierarchy of relaxations based on semidefinite programming~\cite{DLTW08,NPA07} that we will refer to as the QC SDP hierarchy.\footnote{Here ``QC'' stands for ``Quantum Commuting''.} These relaxations follow a similar spirit as e.g. the Lasserre hierarchy in combinatorial optimization~\cite{Laurent03lasserre}, and can be formulated using the language of sums of squares of \emph{non-commutative} polynomials. In contrast with the commutative setting, this leads to a hierarchy that is in general infinite and need not converge at any finite level. 

The limited convergence results that are known for the QC SDP hierarchy involve a formalization of locality for quantum provers which originates in the study of infinite-dimensional systems such as those that arise in quantum field theory. Here the idea is that observations made at different space-time locations should be represented by operators which, although they may act on the same Hilbert space, should nevertheless commute --- a minimal requirement ensuring that the joint outcome of any two measurements made by distinct parties should be well-defined and independent of the order in which the measurements were performed.  

For the case of finite-dimensional systems this seemingly weaker condition is equivalent to the existence of a tensor product representation~\cite{DLTW08}. In contrast, for the case of infinite-dimensional systems the two formulations are not known to be equivalent. This question, known as Tsirelson's problem in quantum information, was recently shown to be equivalent to a host of deep mathematical conjectures~\cite{ScholzW08,JungeNPPSW11}, in particular Connes' embedding conjecture~\cite{Connes76} and Kirchberg's QWEP conjecture~\cite{Kirchberg93qwep}. 
The validity of these conjectures has a direct bearing on our understanding of $\MIPstar$. The QC SDP hierarchy is known to converge to a value called the \emph{field-theoretic value} $\omega^f$ of the game, which is the maximum success probability achievable by commuting strategies of the type described above. 
A positive answer to Tsirelson's conjecture thus implies that $\omega^*=\omega^f$ and both quantities are computable. However, even assuming the conjecture and in spite of strong interest (the use of the first few levels of the hierarchy has proven extremely helpful to study a range of questions in device independence~\cite{BancalSS14data,YangVBSN14robust} and the study of nonlocality~\cite{PalV10I3322}) absolutely no bounds have been obtained on the convergence rate of the hierarchy. It is only known that if a certain technical condition, called a rank loop, holds, then convergence is achieved~\cite{NPA08NJP}; unfortunately the condition is computationally expensive to verify (even for low levels of the hierarchy) and, in general, may not be satisfied at any finite level.

Beyond the obvious limitations for practical applications, these severe computational difficulties are representative of the intrinsic difficulty of working with the model of entangled provers. 
Our work is motivated by this state of affairs: we establish the first quantitative convergence results for the quantum SDP hierarchy. Our main observation is that successive levels of the hierarchy place bounds on the value achievable by provers employing a relaxed notion of strategy in which measurements applied by distinct provers are allowed to \emph{approximately commute}: their commutator is bounded, in operator norm, by a quantity that tends to zero as the number of levels in the hierarchy grows. 

\subsection*{A rounding scheme for the QC SDP hierarchy} \label{sec:results}

Our main technical result is a rounding procedure for the QC SDP hierarchy of semidefinite programs. The procedure maps any feasible solution to the $N$-th level of the hierarchy to a set of measurement operators for the provers that approximately commute. For simplicity we state and prove our results for the case of a single round of interaction with two provers and classical messages only. Extension to multiple provers is straightforward; we expect generalizations to multiple rounds and quantum messages to be possible but leave them for future work.

\begin{definition}
An $(m,\ell)$ strategy for the provers is specified by two sets of $m$ POVMS $\{A_x^a\}_{1\leq a \leq \ell}$ and $\{B_y^b\}_{1\leq b \leq \ell}$ with $\ell$ outcomes each, where $x,y\in\{1,\ldots,m\}$. 

A strategy is said to be $\delta$-AC if for every $x,y,a$ and $b$, $\|A_x^a B_y^b - B_y^b A_x^a\|\leq \delta$, where $\|\cdot\|$ denotes the operator norm.
\end{definition}

Our main theorem on the QC SDP hierarchy can be stated as follows. (We refer to Section~\ref{Sec:def-sdp} for a definition of the hierarchy.)

\begin{theorem}\label{thm:sdp}
Let $G$ be a $2$-prover one-round game with classical messages in which each prover has $\ell$ possible answers, and $\omega_{QCSDP}^N(G)$ the optimum of the $N$-th level of the QC SDP hierarchy. Then there exists a $\delta=O(\ell^2/\sqrt{N})$ and a $\delta$-AC strategy for the provers with success probability $\omega_{QCSDP}^N(G)$ in $G$.\footnote{Due to the approximate commutation of the provers' strategies the success probability of $\delta$-AC strategies may a priori depend on the order in which the measurement operators are applied. In our context the parameter $\delta$ will always be small enough that we can neglect this effect. Moreover, for the particular kind of strategies constructed in our rounding scheme the value will not be affected by the order.}
\end{theorem}

Our result is the first to derive the condition that the \emph{operator norm} of commutators is small. In contrast it is not hard to show that a feasible solution to the first level of the hierarchy already gives rise to measurement operators that exactly satisfy a commutation relation \emph{when evaluated on the state} (corresponding to the zeroth-order vector provided by the hierarchy). While the latter condition can be successfully exploited to give an exact rounding procedure from the first level for the class of XOR games~\cite{CHTW04}, and an approximate rounding for the more general class of unique games~\cite{KRT10}, we do not expect it to be sufficient in general.  In particular, even approximate tightness of the first level of the hierarchy for three-prover games would imply $\EXP=\NEXP$~\cite{Vidick13xor}.  
We will further show that the problem of optimizing over strategies which approximately commute, to within sufficiently small error and in \emph{operator norm}, is $\NEXP$-hard (see Section~\ref{sec:mipdelta} for details). 
%There are, of course, other advantages of a bound on the operator norm of the commutator.  For example, the operator norm is unitarily invariant, and expresses a bound that applies to the entire Hilbert space rather than an small fraction of the Hilbert space.  However, the authors feel that, currently, the above complexity theoretic reasoning is the best justification for studying bounds on the operator norm. 

The proof of Theorem~\ref{thm:sdp} is constructive: starting from any feasible solution to the $N$-th level of the QC SDP hierarchy we construct measurement operators for the provers with pairwise commutators bounded by $\delta$ in operator norm, and which achieve a value in the game that equals the objective value of the $N$-th level SDP. Recall that this SDP has $O(m\ell)^N$ vector variables indexed by strings of length at most $N$ over the formal alphabet $\{P_x^a,Q_y^b\}$ containing a symbol for each possible (question,answer) pair to any of the provers. Our main idea is to introduce a ``graded'' variant of the construction in~\cite{NPA08NJP} (which was used to show convergence under the rank loop constraint). Rather informally, 
%\begin{definition}[Rounding Scheme for the QCSDP hierarchy] \label{roundingscheme}
the rounded measurement operators, $\{\tilde{P}_x^a\}$ for the first prover and $\{\tilde{Q}_y^b\}$ for the second, can be defined as follows: 
\begin{align*}
\tilde{P}_x^a \equiv  \frac{1}{N-1} \sum_{i=1}^{N-1}  \Pileq{i} \Pi_{P_x^a} \Pileq{i}\quad\text{and}\quad
\tilde{Q}_y^b \equiv \frac{1}{N-1} \sum_{j=1}^{N-1}  \Pileq{j} \Pi_{Q_y^b} \Pileq{j}.
\end{align*}
Here $\Pi_{P_x^a}$ and $\Pi_{Q_y^b}$ are projectors as defined in~\cite{NPA08NJP}, i.e. as the projection onto vectors associated with strings ending in the formal label $P_x^a$, $Q_y^b$ of the corresponding operator. The novelty is the introduction of the $\Pileq{i}$, which project onto the subspace spanned by all vectors associated with strings of length at most $i$. Thus $\tilde{P}_x^a$ itself is not a projector, and it gives more weight to vectors indexed by shorter strings.

The intuition behind this rounding scheme is as follows. The winning probability is unchanged because it is determined by the action of the measurement operators on the subspace $\Ima( \Pileq{1})$.  On the other hand, the rounded operators approximately commute in the operator norm because the original operators commuted exactly on the subspace $\Ima(\Pileq{N-1})$, and we have now shifted the weight of the operators so that they are supported on that subspace.  Furthermore, while truncating the operators abruptly at level $N-1$ (by conjugating by $ \Pileq{N-1}$ for example) could result in a large commutator, we perform a ``smooth'' truncation across vectors indexed by strings of increasing length.

\subsection*{Interactive proofs with approximately commuting provers}\label{sec:mipdelta}

Motivated by the rounding procedure ascertained in Theorem~\ref{thm:sdp} we propose a modification of the class $\MIP^*$ in which the assumption that isolated provers must perform perfectly commuting measurements is relaxed to a weaker condition of \emph{approximately commuting} measurements.  
\begin{definition} \label{def::MIPdelt} Let $\MIP^*_\delta(k,c,s)$ be the class of promise problems $(L_{yes},L_{no})$ that can be decided by an interactive proof system in which the verifier exchanges a single round of classical messages with $k$ quantum provers $P_1,\ldots,P_k$ and such that:

\begin{itemize}[itemsep=0pt, topsep=5pt, partopsep=5pt]
\item If the input $x\in L_{yes}$ then there exists a perfectly commuting strategy for the provers that is accepted with probability at least $c$,
\item If $x\in L_{no}$ then any $\delta$-AC strategy is accepted with probability at most $s$. 
\end{itemize}
\end{definition}

Note that the definition of $\MIP^*_\delta$ requires the completeness property to be satisfied with perfectly commuting provers; indeed we would find it artificial to seek protocols for which optimal strategies in the ``honest'' case would be required to depart from the commutation condition. Instead, only the soundness condition is relaxed by giving \emph{more} power to the provers, who are now allowed to apply any ``approximately commuting'' strategy. The ``approximately'' is quantified by the parameter $\delta$,\footnote{As a first approximation the reader may think of $\delta$ as a parameter that is inverse exponential in the input length $|x|$. In terms of games, this corresponds to $\delta$ being inverse polynomial in the number of questions in the game, which is arguably the most natural setting of parameters.} and for any $\delta'\leq \delta$ the inclusions $\MIP^*_\delta \subseteq \MIP^*_{\delta'} \subseteq \MIP^*$ trivially hold. It is important to keep in mind that while $\delta$ can be a function of the size of the protocol it must be independent of the dimension of the provers' operators, which is unrestricted. 

$\delta$-AC strategies were previously considered by Ozawa~\cite{Ozawa13commuting} in connection with Tsirelson's problem. Ozawa proposes a conjecture, the ``Strong Kirchberg Conjecture (I)'', which if true implies the equality $\MIP^* = \cup_{\delta>0}\, \MIP_\delta^*$. We state and discuss the conjecture further as Conjecture~\ref{conj:ozawa} below. Unfortunately the conjecture seems well beyond the reach of current techniques (Ozawa himself formulates doubts as to its  validity). However, in our context less stringent formulations of the conjecture would still imply conclusive results relating $\MIP^*_\delta$ to $\MIP^*$; we discuss such variants in Section~\ref{sec:discussion}. 

Further motivation for the definition of $\MIP^*_\delta$ may be found by thinking operationally --- with e.g. cryptographic applications in mind, how does one ascertain that ``isolated'' provers indeed apply commuting  measurements? The usual line of reasoning applies the laws of quantum mechanics and special relativity to derive the tensor product structure from space-time separation. However, not only is strict isolation virtually impossible to enforce in all but the simplest experimental scenarios, but the implication ``separation $\implies$ tensor product'' may itself be subject to questioning --- in particular it may not be a testable prediction, at least not to precision that exceeds the number of measurements, or observations, performed. Relaxations of the tensor product condition have been previously considered in the context of device-independent cryptography; for instance Silman et al.~\cite{SPM13} require that the joint measurement performed by two isolated devices be close, in operator norm, to a tensor product measurement. Our approximate commutation condition imposes a weaker requirement, and thus our convergence results on the hierarchy also apply to their setting; we discuss this in more detail in Section~\ref{subsec::SPM}. 

Theorem~\ref{thm:sdp} can be interpreted as evidence that the hierarchy converges at a polynomial rate to the maximum success probability for $\AMIP^*$ provers. More formally, it implies the inclusion $\MIP^*_\delta \subseteq \TIME(\expt(\exp(\poly)/\delta^2))$ for any $\delta>0$, thereby justifying our claim that the class $\MIP^*_\delta$ is computationally bounded. This stands in stark contrast with $\MIP^* = \MIP_0^*$, for which no quantitative upper bound is known  (We note, however, that $\MIP^*$ is known to be recursively enumerable).

Having shown that the new class has ``reasonable'' complexity, it is natural to ask whether the additional power granted to the provers might actually make the class trivial --- could provers that are $\delta$-AC be no more useful than a single quantum prover, even for very small $\delta$? The following theorem shows that this is not the case. 

\begin{theorem}\label{thm:nexp}
Every language in $\NEXP$ has a $2$-prover $\MIP^*$ protocol in which completeness $1$ holds with classical provers and soundness $2^{-\poly}$ holds even against provers thate are $2^{-\poly}$-AC. More formally,
$$ \NEXP \subseteq \bigcup_{p,q\in\poly} \MIP_{2^{-q}}^*(2,1,1-2^{-p}).$$
\end{theorem}

Theorem~\ref{thm:nexp} provides a direct analogue of the same lower bound for $\MIP^*$~\cite{ikm09}, and is proven using the same technique. We conjecture that the inclusion $\NEXP\subseteq \MIP^*_{2^{-\poly}}(3,1,2/3)$ also holds, and that this can be derived by a careful extension of the results in~\cite{IV12,Vidick13xor}.

\paragraph{Organization.}
Section~\ref{sec:prelim} contains some preliminaries, including the definition of $\AMIP^*$ and the QC SDP hierarchy. In Section~\ref{commutatorbound} we present and analyze our rounding scheme for the hierarchy. In Section~\ref{sec:nexp} we prove our lower bound on $\MIP^*$. We end in Section~\ref{sec:discussion} with a discussion of the relevance of the study of $\AMIP^*$ for that of $\MIP^*$ and closely related results from the mathematical literature.

\section{Preliminaries}\label{sec:prelim}

\paragraph{Notation.} Given an integer $N$, we use $[N]$ to refer to the set $\{1,\ldots,N\}$. We use the symbol $\equiv$ to mark a definition. $\mathfrak{S}_k$ is the set of all permutations $\sigma:[k]\to[k]$. $\poly$ denotes the set of all polynomials. We write $\|X\|$ for the operator norm of a matrix $X$. Given matrices $A,B$, $[A,B]\equiv AB-BA$ denotes their commutator.

\subsection{Approximately commuting provers}\label{interactiveproofs}

In this section we define the class $\MIP_\delta^*$ for $\delta>0$ and state some basic properties. We assume the reader is familiar with quantum interactive proof systems and the definition of the class $\MIP^*$; we refer to e.g.~\cite{KobMat03JCSS} for details. Throughout we will use $k$ to denote the number of provers and $c,s$ the completeness and soundness parameters. Although one could define the class more generally, we restrict our attention to protocols involving a single round of interaction between the verifier and the provers. As is customary we will also call such one-round protocols \emph{games}. 

\begin{definition} 
A $k$-prover game $G$ is specified by the following: integers $Q_1,\ldots,Q_k$, representing the number of possible questions to each prover, and a distribution $\pi$ on $[Q_1]\times \cdots \times [Q_k]$; integers $A_1,\ldots,A_k$, representing the number of possible answers from each prover; a mapping $V:([Q_1]\times \cdots \times [Q_k])\times([A_1]\times \cdots \times [A_k])\to\{0,1\}$ representing the referee's acceptance criterion. 
\end{definition}

Next we introduce a notion of \emph{approximately commuting} strategies for the provers. 

\begin{definition}
Given a game $G$, a \emph{strategy} for the provers consists of the following:
\begin{itemize}
\item A finite-dimensional Hilbert space $\mathcal{H}$,
\item For every $i\in [k]$ and $q\in [Q_i]$, a POVM $\{(A^{(i)})_q^a\}_{a\in [A_i]}$, where each $(A^{(i)})_q^a \in \pos{\mathcal{H}}$ and $\sum_a (A^{(i)})_q^a = \Id$,
\item A density matrix $\rho \in \density{\mathcal{H}}$. 
\end{itemize}
For any $\delta>0$ we say that the strategy $((A^{(1)})_{q}^{a}, \ldots,(A^{(k)})_{q'}^{a'},\rho)$ is $\delta$-AC if $\|[(A^{(i)})_{q}^{a},(A^{(j)})_{q'}^{a'}]\|\leq \delta$ for every $i\neq j \in [k]$ and $q\in [Q_i]$, $q'\in [Q_j]$, $a\in [A_i]$ and $a'\in[A_j]$. 
\end{definition}

Finally we define the success probability, or \emph{value}, achieved by a given strategy in a game. 

\begin{definition}
Let $G$ be a game and $((A^{(1)})_{q}^{a}, \ldots,(A^{(k)})_{q'}^{a'},\rho)$ a strategy in $G$. The strategy's value is defined as
$$ \omega^*\big( ((A^{(1)})_{q}^{a}, \ldots,(A^{(k)})_{q'}^{a'},\rho);G\big)\,:=\,\max_{\sigma\in\mathfrak{S}_k} \left | \sum_{q_i\in Q_i} \pi(q_i) \sum_{a_1,\ldots, a_k} V(a_i,q_i) \Tr\big( (A^{(\sigma(1))})_{q_{\sigma(1)}}^{a_{\sigma(1)}} \cdots (A^{(\sigma(k))})_{q_{\sigma(k)}}^{a_{\sigma(k)}} \rho\big)  \right |.$$
The \emph{$\delta$-AC value} $\omega_\delta^*(G)$ of the game $G$ is the supremum over all $\delta$-AC strategies of the strategy's value in $G$. 
\end{definition}

In the above definition the introduction of the supremum over all permutations of the provers amounts to allowing the provers to choose the order in which their respective POVM are applied to the shared state (the ordering should be the same throughout, independently of the questions asked). Since POVM elements applied by distinct provers do not necessarily commute the choice of ordering may affect a strategy's value. Nevertheless, for $\delta$-AC strategies it is easy to see that any two orderings will result in values that differ by at most $k^2\delta$; for our purposes the parameter $\delta$ will always be small enough that different choices of orderings would not matter and we will mostly ignore this issue throughout. Since we consider only finite-dimensional strategies, for $\delta=0$ the value $\omega^*_0(G)$ reduces to what is usually called the \emph{entangled value} $\omega^*(G)$, corresponding to strategies that are perfectly commuting, or equivalently strategies that can be put in tensor product form. 

Having introduced games, strategies, and values, we are ready to define the class $\MIP_\delta^*$. 

\begin{definition}
Let $\delta,c,s:\N\to [0,1]$ be computable functions and $k\in\poly$. A language $L$ is in ${\MIP_\delta^*(k, c, s)}$ if and only if there exists a polynomial-time computable mapping from inputs $x\in\{0,1\}^*$ to $k$-prover games $G_x$ such that:
\begin{itemize}
\item In the game $G_x$, the distribution $\pi$ can be sampled in time polynomial in $|x|$, and the predicate $V$ can be computed in time polynomial in $|x|$,
\item (Completeness) If ${x \in L}$ then $\omega^*(G_x)\geq c$, i.e. there exist a $k$-prover strategy that is $0$-AC and has value at least $c$ in $G_x$, 
 \item (Soundness) If ${x \not\in L}$ then $\omega_\delta^*(G_x)\leq s$, i.e. every $\delta$-AC $k$-prover strategy has value at most $s$ in $G_x$. 
\end{itemize}
\label{dem:mip-delta}
Throughout this work we use $\poly$ to represent any polynomial in $|x|$, or equivalently, any polynomial in the length of the messages passed in the protocol.

\end{definition}

Since for every $\delta < \delta'$ a $\delta'$-AC strategy is also $\delta$-AC, it follows that $\MIP_{\delta'}^*(k, c, s)\subseteq  \MIP_\delta^*(k, c, s)$.
A choice of parameter $\delta$ that is inverse exponential in the input length seems to be the most natural, and we define 
$$\AMIP^*(k, c, s) \,:=\, \cup_{\delta\in 2^{-\poly}} \MIP_{\delta}^*(k, c, s).$$ 

We end this section with a simple claim on $\delta$-AC strategies that is well-known to hold for the case of perfectly commuting strategies: up to a small loss in the commutation parameter we may without loss of generality restrict ourselves to strategies that apply projective measurements. 

\begin{claim}\label{claim:projective}
Let $(A_q^a,B_{q'}^{a'},\rho)$ be a $\delta$-AC strategy with success probability $p$ in a certain game. Then there exists a $(|A| |A'| \delta)$-AC strategy in which all POVM are projective and that achieves the same success probability $p$. 
\end{claim}

\begin{proof}
We can make $(A_q^a,B_{q'}^{a'},\rho)$ (which is defined on some Hilbert space $\mathcal{H}$) into a projective strategy $(\tilde{A}_q^a,\tilde{B}_{q'}^{a'}, \tilde{\rho})$ on an extended Hilbert space $\mathcal{H}' \equiv \mathcal{H} \otimes \mathbb{C}^{|A|} \otimes \mathbb{C}^{|A'|}$  by extending $\rho$ to $\rho \otimes \ket{0}_A \bra{0}\otimes \ket{0}_B \bra{0}$, and defining the norm preserving maps :

\begin{align*}
&U^q_{A}: \ket{\psi}\ket{0}_A \ket{0}_B \to\sum_a (\sqrt{A_q^a}\ket{\psi}) \ket{a}_A \ket{0}_B \\
&\text{and}\\
&U^{q'}_{B}: \ket{\psi}\ket{0}_A \ket{0}_B \to\sum_{a'} (\sqrt{B^{a'}_{q'}}\ket{\psi}) \ket{0}_A \ket{a'}_B 
\end{align*}

Since the maps are norm-preserving they can be extended to unitary maps $\tilde{U}^q_{A} \otimes I_{B}$ and $\tilde{U}^{q'}_{B}\otimes I_{A}$ respectively on $\mathcal{H} \otimes \mathbb{C}^{|A|} \otimes \mathbb{C}^{|A'|}$ (note that the unitary for each prover acts as the identity on the ancilla for each other prover).  We now define the new POVM operators as 

\begin{align*}
& \tilde{A}_q^a \equiv \left (\tilde{U}^q_{A} \otimes I_{B} \right ) \left ( I_{\mathcal{H}} \otimes \ket{a}_A\bra{a} \otimes I_{B} \right) \left (\tilde{U}^q_{A} \otimes I_{B} \right)^{\dagger} \\
& \text{and}\\
& \tilde{B}_{q'}^{a'} \equiv (\tilde{U}^{q'}_{B}\otimes I_{A}) \left( I_{\mathcal{H}} \otimes I_{A} \otimes \ket{a'}_B\bra{a'} \right ) (\tilde{U}^{q'}_{B}\otimes I_{A})^{\dagger}
\end{align*}
The new operators now form projective POVM strategies, and the transformation clearly preserves the strategy's success probability.  Since different provers act on distinct ancilla qubits, and as the identity on the ancillas for all other provers, we see that:
$ [\tilde{A}_q^a,\tilde{B}_{q'}^{a'}] \leq |A| |A'| \delta.$
\end{proof}

\subsection{The QC SDP Hierarchy}\label{Sec:def-sdp}

Fix a two-prover game $G$. Let $X=[Q_1]$ (resp. $Y=[Q_2]$) be the first (resp. second) prover's input alphabet, and $A=[A_1]$ (resp. $B=[A_2]$) the first (resp. second) prover's answer alphabet.  Let $V: X \times Y \times A \times B \to \{ 0, 1 \}$ be the referee's decision predicate, and $\mu: X \times Y \to [0,1]$ the distribution on inputs.  Consider the alphabet of formal symbols $\mathcal{A} \equiv \{P^a_x:  \forall x, a   \} \cup \{Q_y^b:   \forall y, b  \}$,\footnote{Although this is a formal alphabet, $P_x^a$ (resp. $Q_y^b$) is meant to represent the first (resp. second) provers' POVM element associated with input $x$ (resp. $y$) and answer $a$ (resp. $b$).} and let $W_m \equiv \cup_{i = 1}^m \mathcal{A}^i \cup \{\phi\}$ be the set of all words of length at most $m$ on the alphabet $\mathcal{A}$  (here $\phi$ is a formal symbol representing the empty string).  The $N^{th}$ level of QC SDP hierarchy for $G$ defines an optimization problem over the space of positive semidefinite matrices $\Gamma_N \in \mathbb{C}^{|W_N| \times |W_N|}$ with entries $\Gamma^N_{s,t}$ indexed by words $s,t \in W_N$.  As in \cite{pna08}, we will let $\{|v_s \rangle:  s \in W_N  \}$ be vectors in $\mathbb{C}^{|W_N|}$ such that $\Gamma^N_{s,t} = \langle v_s | v_t \rangle$. We can find such vectors by computing, for example, the Cholesky decomposition of $\Gamma^N$, which can be done in time polynomial in $|W_N|$.

\begin{definition}\label{def:qcsdp}
The $N^{th}$ level of QC SDP hierarchy is defined to be the following optimization problem:
\begin{align}
&\text{maximize  }  \sum_{(x,y,a,b)} \mu(x,y) \Gamma^N_{P_x^a, Q_y^b} \ V(x,y,a,b) \notag\\
&\notag\\
&\text{subject to:  } \notag\\
&\Gamma^N \succeq 0 \notag\\
&\Gamma^N_{\phi, \phi} = 1 \notag\\
&\forall C \in \mathcal{A}, \forall s,t \in W_{N-1},  \Gamma^N_{sC, t} = \Gamma^N_{s, Ct} \notag\\
&\forall P_x^a, Q_y^b \in \mathcal{A}, \forall s,t \in W_{N-1}, \Gamma^N_{sP_x^a, Q_y^bt} = \Gamma^N_{sQ_y^b, P_x^at}  &\label{commutationconstraint} \\
&\forall x \in X, s \in W_{N-1}, t \in W_N, \sum_{a \in A} \Gamma^N_{sP_x^a, t} = \Gamma^N_{s, t}\label{sumequalityinnera} \\
&\forall y \in Y, s \in W_{N-1}, t \in W_N,  \sum_{b \in B} \Gamma^N_{sQ_y^b, t} = \Gamma^N_{s, t}  \notag\\
&\forall x \in X, \forall s,t \in W_{N-1}, \text{and for } a \neq a',  \Gamma^N_{sP_x^a, P_x^{a'}t} = 0  \label{aorthog}\\
&\forall y \in Y, \forall s,t \in W_{N-1}, \text{and for } b \neq b',  \Gamma^N_{sQ_y^b,Q_y^{b'} t} = 0  \nonumber
\end{align}
\end{definition}

Note that, in order to make the constraints intuitive, we use the non-standard notation that, whenever as vector $|v_s \rangle$ is transposed, the result of the transpose is written as $\bra{v_s}:= (\ket{v_{s^\dagger}})^\dagger$, where $s^{\dagger}$ is the string written in the reverse order.  That is, we use the convention that transposing vectors also reverses the order of their labels. 

From here on we let $\Gamma^N_{s,t} = \langle v_s | v_t \rangle$ represent an optimal solution to the $N^{th}$ level of the QC SDP hierarchy.  By definition,

\begin{align*}
\omega_{QCSDP}^N(G) = \sum_{(x,y,a,b)} \mu(x,y) \Gamma^N_{P_x^a, Q_y^{b}} \ V(x,y,a,b) = \sum_{(x,y,a,b)} \mu(x,y) \langle v_{P_x^a} | v_{Q_y^b}  \rangle  V(x,y,a,b).
\end{align*}  

\begin{definition}
Let $V_j \equiv \text{Span}  \{ |v_s \rangle : s \in W_j \}  $ denote the vector space spanned by all the vectors with labels of length $\leq j$.  Note that, for a solution to the $N^{th}$ level of the QC SDP hierarchy, $V_N$ is the entire space spanned by all the vectors in the Cholesky decomposition of $\Gamma^N$.
\end{definition}

\begin{observation}  \label{vector sum}

\begin{align}  
\forall x \in X, s \in W_{N-1}, t \in W_N, \sum_{a \in A} \langle v_{sP_x^a} | = \langle v_{s} |  \label{sumequalitya} \\
\forall y \in Y, s \in W_{N-1}, t \in W_N, \sum_{b \in B} \langle v_{sQ_y^b} |= \langle v_{s} |    \label{sumequalityb}
\end{align}
\end{observation}

\begin{proof}
We give the proof of equation \eqref{sumequalitya}.  The proof for equation \eqref{sumequalityb} is completely analogous.  Consider the vector $| z \rangle \equiv  |v_{s} \rangle - \sum_{a \in A} | v_{sP_x^a} \rangle$.  By definition we have that $| z \rangle \in V_N$.  On the other hand, for every vector $|v_t \rangle$ (for any $t \in W_N$), it follows from equation \eqref{sumequalityinnera} that $\langle z | v_t \rangle = \langle v_{s} | v_t \rangle  - \sum_{a \in A} \langle v_{sP_x^a}  | v_t \rangle = 0$.  Thus we must have $| z \rangle = 0$, and the claim follows.
\end{proof}

\begin{definition}
As in \cite{pna08}, for each $P_x^a \in \mathcal{A}$, let $\Pi_{P_x^a}$ denote the projector onto $\text{Span}\{ |v_{P_x^a s} \rangle : s \in W_{N-1} \}$.  Similarly, for each $Q_y^b \in \mathcal{A}$, let $\Pi_{ Q_y^b}$ denote the projector onto $\text{Span}\{ |v_{Q_y^b s} \rangle : s \in W_{N-1} \}$.  
\end{definition} 

\begin{observation} \label{projectoraction}
Note that, as observed in \cite{pna08}, for each $P_x^{a'} \in \mathcal{A}$, and for all $s \in W_{N-1}$ we have that: 
\begin{align}
&\Pi_{P_x^{a'}} | v_{s} \rangle = \Pi_{P_x^{a'}} \sum_{a \in A} | v_{P_x^as} \rangle =  \Pi_{P_x^{a'}} | v_{P_x^{a'}s} \rangle = | v_{P_x^{a'}s} \rangle
\end{align}
   This follows from equation \eqref{aorthog} and  Observation \ref{vector sum}.  The analogous statement holds for $Q_y^{b'} \in \mathcal{A}$, and $\Pi_{Q_y^{b'}}$.
\end{observation}

\begin{definition}
For each  $j \leq N$, let $\Pileq{j}$ denote the orthogonal projector onto $V_j$, and $\Pileq{j}^{\perp} \equiv I - \Pileq{j}$.  
\end{definition}

\subsection{Useful identities}  \label{projectoridentitiesappendix}

The following identities involving the projection operators defined in the previous section will be used in the analysis of our rounding scheme in Section~\ref{commutatorbound}. 

\begin{proposition}  \label{projectorlevelidentities}
$\forall i, j \in [N]$:
\begin{align}
& \Pileq{i} \Pileq{j} = \Pileq{j} \Pileq{i} = \Pileq{\min(i, j)} \label{standardinclusion} \\
& \Pileq{i}^{\perp} \Pileq{j}^{\perp} = \Pileq{j}^{\perp} \Pileq{i}^{\perp} = \Pileq{\max(i, j)}^{\perp}.  \label{perpinclusion}
\end{align}
Furthermore, for $i \geq j \in [N]$, 
\begin{align}
&\Pileq{j} \Pileq{i}^{\perp} = \Pileq{i}^{\perp} \Pileq{j} = 0,  \label{standardorthog}
\end{align}
\end{proposition}

\begin{proof}
All three equations follow trivially from the definition of $\Pileq{i}$ as the orthogonal projector on $V_i$ and the inclusion $V_j  \subseteq V_i$ for $j\leq i$. 
\end{proof}

\begin{proposition} \label{commutation-cancellation identity}
$ \forall (x,y, a,b) \in X \times Y \times A \times B$, and $\forall i,j < N$,
$$\Pileq{i} \Pi_{P_x^a}  \Pi_{Q_y^b} \Pileq{j}  = \Pileq{i} \Pi_{Q_y^b}  \Pi_{P_x^a} \Pileq{j}.$$
\end{proposition}

\begin{proof}
Consider any two vectors $|z  \rangle, |w \rangle \in V_N$.  By definition, we have that $\Pileq{i} |z \rangle \in \text{Im} \left ( \Pileq{i} \right ) \equiv V_i =  \text{Span}  \{ |v_s \rangle : s \in W_i \} $, so we can expand them in this vector space:  $\Pileq{i} |z \rangle \equiv \sum_{ s \in W_i } \lambda_s |v_s \rangle$. Similarly, $\Pileq{j} |w \rangle \equiv \sum_{ t \in W_j } \gamma_t |v_t \rangle$.  So,
    \begin{align}
\langle z | \Pileq{i} \Pi_{P_x^a}  \Pi_{Q_y^b} \Pileq{j} | w \rangle & = \Big (  \sum_{ s \in W_i } \lambda_s^* \langle v_{s^{\dagger}} |  \Big )  \Pi_{P_x^a}  \Pi_{Q_y^b}  \Big (\sum_{ t \in W_j } \gamma_t |v_t \rangle  \Big ) = \sum_{ s \in W_i }   \sum_{ t \in W_j }  \lambda_s^* \gamma_t  \langle v_{s^{\dagger}} | \Pi_{P_x^a}  \Pi_{Q_y^b} |v_t \rangle \notag\\
  & = \sum_{ s \in W_i }   \sum_{ t \in W_j }  \lambda_s^* \gamma_t  \langle v_{s^{\dagger} P_x^a} |v_{Q_y^bt} \rangle = \sum_{ s \in W_i }   \sum_{ t \in W_j }  \lambda_s^* \gamma_t  \langle v_{s^{\dagger} Q_y^b} |v_{P_x^at} \rangle \notag \\
  & = \sum_{ s \in W_i }   \sum_{ t \in W_j }  \lambda_s^* \gamma_t  \langle v_{s^{\dagger} } | \Pi_{Q_y^b} \Pi_{P_x^a}  |v_{t} \rangle  = \Big (  \sum_{ s \in W_i } \lambda_s^* \langle v_{s^{\dagger}} |  \Big )    \Pi_{Q_y^b} \Pi_{P_x^a}  \Big (\sum_{ t \in W_j } \gamma_t |v_t \rangle  \Big ) \notag \\
  & = \langle z |\Pileq{i}   \Pi_{Q_y^b} \Pi_{P_x^a} \Pileq{j} | w \rangle \label{perpswapexpansion}
  \end{align}
    In \eqref{perpswapexpansion}, since we know that $s \in W_i$, $t\in W_j$ and $W_i,W_j \subseteq W_{N-1}$, the third equality follows by Observation \ref{projectoraction}, the fourth equality follows by equation  \eqref{commutationconstraint}, and the fifth equality follows by Observation \ref{projectoraction} again.
    
  Since this holds for arbitrary $|z  \rangle, |w \rangle$, it follows that $\Pileq{i} \Pi_{P_x^a}  \Pi_{Q_y^b} \Pileq{j}  = \Pileq{i} \Pi_{Q_y^b}  \Pi_{P_x^a} \Pileq{j}$, as claimed.  
\end{proof}

\begin{proposition}  \label{projectoroneshift}
For any $j < N$, and any $P_x^a \in \mathcal{A}$, or $Q_y^b \in \mathcal{A}$ we have that 

\begin{equation*}
\text{Im} \big(\Pi_{P_x^a} \Pileq{j} \big) \subseteq \text{Im} \big(\Pileq{j+1} \big)\qquad\text{and}\qquad
\text{Im} \big(\Pi_{Q_y^b} \Pileq{j} \big) \subseteq \text{Im} \big(\Pileq{j+1} \big) 
\end{equation*}

\end{proposition}

\begin{proof}
Let $|z \rangle$ be an arbitrary vector in $V_N$.  By definition, $\Pileq{j} |z \rangle \in \text{Im} \left (\Pileq{j} \right) \equiv V_j = \text{Span}  \{ |v_s \rangle : s \in W_{j} \} $.  So, there exist coefficients $\lambda_s \in \mathbb{C}$ such that $\Pileq{j} |z \rangle = \sum_{s \in W_j} \lambda _s |v_s \rangle$.  Now, by invoking Observation \ref{projectoraction} we see that 

\begin{align*}
& \Pi_{P_x^a} \Pileq{j} | z \rangle =\sum_{s \in W_j} \lambda _s  \Pi_{P_x^a}  |v_s \rangle =  \sum_{s \in W_j} \lambda _s    |v_{P_x^as} \rangle \in    \text{Im} \left (\Pileq{j+1} \right) \equiv V_{j+1} = \text{Span}   \{ |v_s \rangle : s \in W_{j+1} \}\\
& \Pi_{Q_y^b} \Pileq{j} | z \rangle =\sum_{s \in W_j} \lambda _s  \Pi_{Q_y^b}  |v_s \rangle =  \sum_{s \in W_j} \lambda _s    |v_{Q_y^bs} \rangle \in    \text{Im} \left (\Pileq{j+1} \right) \equiv  V_{j+1} = \text{Span}  \{ |v_s \rangle : s \in W_{j+1} \} 
\end{align*}

Since this is true for arbitrary $| z \rangle$, the desired result follows.
\end{proof}

\subsection{Some bounds} \label{sec:lowboundprelim}

In this section we collect a few identities that will be useful in the proof of Theorem \ref{thm:nexp}. We first note the bound 
\beq\label{eq:com-bound}
 \| [A,B^r]\| \leq 2\|B\|^{1-r} \|[A,B]\|^r,
\eeq
valid for $A,B\geq 0$ and $0\leq r \leq 1$, that will be useful in our analysis. See Problem X.5.3 in~\cite{bhatia} for a tighter bound from which~\eqref{eq:com-bound} follows.

\begin{claim}\label{claim:sq-bound}
For $i=1,\ldots,M$ let $A_i,B_i \geq 0$ be such that $(A_i)$ and $(B_i)$ are $\delta$-AC, and $\sum_i \Tr(A_i\sqrt{B_i} \rho\sqrt{B_i}) \geq 1-\eps$. Then
\beq\label{eq:sq-bound}
\sum_i \Tr\Big( \big(A_i^{1/2} - B_i^{1/2}\Big)^2 \rho\big) \,\leq\, 2\eps + O\big(\delta^{1/8}M\big),
\eeq
and 
\beq\label{eq:sq-bound-trace}
 \Big\|\sum_i A_i^{1/2}\rho A_i^{1/2} - \sum_i B_i^{1/2}\rho B_i^{1/2}\Big\|_1 \,\leq\, 2\sqrt{2\eps} + O\big(\delta^{1/16}M^{1/2}\big),
\eeq
\end{claim}

\begin{proof}
We first evaluate
\begin{align}
\sum_i \Tr\big(A_i^{1/2}B_i^{1/2}\rho\big) &= \sum_i \Tr\big(A_i^{1/4}B_i^{1/2}A_i^{1/4}\rho\big) - \sum_i \Tr\big([A_i^{1/4},B_i^{1/2}]A_i^{1/4}\rho\big)\notag\\
&\geq  \sum_i \Tr\big(A_i^{1/4}B_i A_i^{1/4}\rho\big) -2\delta^{1/8} M\notag\\
&=  \sum_i \Tr\big(A_i^{1/2}B_i \rho\big) -   \sum_i \Tr\big(A_i^{1/4}[A_i^{1/4},B_i]\rho\big) -2\delta^{1/8} M\notag\\
&\geq \sum_i \Tr\big(B_i^{1/2}A_i^{1/2}B_i^{1/2} \rho\big) -   \sum_i \Tr\big([B_i^{1/2},A_i^{1/2}]B_i^{1/2}\rho\big)-4\delta^{1/8} M\notag\\
&\geq \sum_i \Tr\big(B_i^{1/2}A_i B_i^{1/2} \rho\big) -  6\delta^{1/8} M\label{eq:com-bound-1}
%&\geq \sum_i \Tr\big(A_i B_i\rho\big) -  8\delta^{1/8} M,
\end{align}
where we repeatedly used the bound~\eqref{eq:com-bound}. 
To obtain~\eqref{eq:sq-bound} it suffices to expand the square in~\eqref{eq:sq-bound} and use the assumption $\sum_i \Tr(A_i\sqrt{B_i} \rho\sqrt{B_i}) \geq 1-\eps$ together with~\eqref{eq:com-bound-1}. Finally,~\eqref{eq:sq-bound-trace} follows easily from~\eqref{eq:sq-bound} (see e.g. Claim~36 in~\cite{iv13arxiv}).
\end{proof}

\section{A rounding scheme for approximately commuting provers} 
\label{commutatorbound}

 We introduce a rounding scheme for the QC SDP hierarchy which, given the optimal $N^{th}$-level QC SDP solution for a certain game $G$, constructs an $O\left ( \frac{ A^2}{\sqrt{N} }  \right )$-AC strategy for the provers (here $A$ is the number of possible answers in $G$, which for simplicity we assume to be the same for each prover). The resulting strategy for $G$ has value equal to the value of $N^{th}$ level of the QC SDP hierarchy, which we denote $\omega_{QCSDP}^N(G)$ (see Definition~\ref{def:qcsdp} below for a precise definition).  To the best of our knowledge this is the first proposal of a rounding scheme for the QC SDP hierarchy for which one is able to provide any quantitative error estimate whatsoever.
    
 In \cite{pna08} and \cite{DLTW08} it is shown that $\omega_{QCSDP}^N(G)$ is an upper bound on the value of $0$-AC strategies, that is, $\omega_{QCSDP}^N(G) \geq \omega^*(G)$.  Our rounding result implies that for all $\delta = O\left ( \frac{ |A|^2}{\sqrt{N} }  \right )$  the quantity $\omega_{QCSDP}^N(G)$ is also a \emph{lower} bound on the optimal success probability achievable by any $\delta$-AC strategy.  This additional result allows us to place an upper bound on the complexity class $\MIP_\delta^*$ introduced in Section~\ref{interactiveproofs}. Precisely, we obtain the following:

\begin{theorem}\label{thm:ub}
For any $\delta>0$, $k\in\N$ and $c,s:\N\to[0,1]$ such that $c-s =\Omega( 2^{-\poly})$ it holds that $\MIP^*_\delta(k,c,s) \subseteq \TIME(\expt(\exp(\poly)/\delta^2))$. Furthermore, the upper bound can be brought down to $\TIME(\expt(\poly/\delta^2))$ when considering only protocols with constant answer size. 
\end{theorem}

Combining Theorem~\ref{thm:ub} with Theorem~\ref{thm:nexp} we obtain that for any constant $k$ it holds that 
$$ \NEXP \,\subseteq\, \bigcup_{p\in\poly} \AMIP^*(2,1,1-2^{-p}) \,\subseteq\, \TIME\big(2^{2^{\poly}}).$$

\subsection{Rounding Scheme}\label{sec:rounding}

In this section we introduce a rounding scheme for the QC SDP hierarchy.  First we  briefly argue that the most natural rounding scheme suggested by the definition of the hierarchy, which was first proposed in \cite{NPA07}, is actually not the best for our purposes.  In \cite{NPA08NJP} it is proposed that any solution, $\Gamma^N$, to the $N^{th}$ level of the QC SDP hierarchy be rounded to a strategy consisting of state $\rho \equiv  |v_{\phi} \rangle \langle v_{\phi}|$, and projective measurement operators $\Pi_{P_x^a}$ for the first prover and $\Pi_{Q_y^b}$ for the second.  It is further proved that, assuming a technical condition called the ``rank loop" condition, this rounded strategy gives valid POVMs for the two provers, and that those POVMs are exactly commuting ($[\Pi_{P_x^a},  \Pi_{Q_y^b}] = 0$).  Unfortunately, the ``rank loop" condition is computationally difficult to verify, and in general it may not hold at any level of the hierarchy.  Even without assuming the ``rank loop" condition, it is true that, for all $j<N$, $ \Pileq{j}  \Pi_{P_x^a}  \Pi_{Q_y^b} \Pileq{j}  = \Pileq{j} \Pi_{Q_y^b}  \Pi_{P_x^a} \Pileq{j}$  (see Proposition \ref{commutation-cancellation identity}).  However, while this tells us that $[\Pi_{P_x^a},  \Pi_{Q_y^b}] = 0$ \emph{exactly} when restricted to the space $V_{N-1}$, it is hard to control the size of $\| [\Pi_{P_x^a},  \Pi_{Q_y^b}]  \|$ on the space $V_{N-1}^{\perp} \equiv \text{Im} \left (\Pileq{N-1}^{\perp} \ \right )$ without making additional assumptions about the structure of $G$, etc.  Furthermore, when using this rounding scheme, it is not clear that there is any quantitative benefit from increasing the number of levels $N$ of the QC SDP hierarchy.

We introduce a rounding scheme which will ultimately allow us to control the operator norm of commutators of the rounded strategy on the entire space $V_N$, without making any assumptions whatsoever about the structure of $G$.

\begin{definition}[Rounding Scheme for the QC SDP hierarchy] \label{roundingscheme}
Fix probability distributions $\{p_i\}_{i=0}^N$, and $\{q_j\}_{j=1}^N$.  In what follows we will assume that $p_0 = q_0 = p_N = q_N = 0$.  Given a solution $\Gamma^N$ to the $N^{th}$ level of the QC SDP hierarchy for $G$, the probability distributions $p_i$ and $q_j$ specify a rounding scheme as follows.
The state shared by the provers is $\rho \equiv  |v_{\phi} \rangle \langle v_{\phi}|$.  Their measurement operators, $\{\tilde{P}_x^a\}$ for the first prover and $\{\tilde{Q}_y^b\}$ for the second, are defined as 
\begin{align}
\tilde{P}_x^a \equiv \sum_i p_i \Pileq{i} \Pi_{P_x^a}\Pileq{i},\qquad
\tilde{P}_x^{garbage} \equiv I - \sum_{a \in \mathcal{A}} \tilde{P}_x^a \nonumber \\
\tilde{Q}_y^b \equiv \sum_j q_j \Pileq{j} \Pi_{Q_y^b}\Pileq{j},\qquad
\tilde{Q}_y^{garbage} \equiv I - \sum_{b \in \mathcal{B}} \tilde{Q}_y^b \nonumber
\end{align}
\end{definition}

We first verify that the rounding scheme defined in Definition~\ref{roundingscheme} defines valid POVM measurements, and leads to a strategy whose value in $G$  exactly matches $\omega_{QCSDP}^N(G)$. (As we defined it, it may seem that the strategy sometimes outputs a symbol ``garbage''. As we show below this has probability $0$, and we can safely ignore the event.)

\begin{claim}\label{claim:strat}
The strategy defined in Definition~\ref{roundingscheme} has value exactly $\omega_{QCSDP}^N(G)$ in $G$.
\end{claim}

\begin{proof}
First we note that the $\tilde{P}_x$ and $\tilde{Q}_y$ define valid POVM. Indeed, using that each $ \Pileq{i}$, $ \Pi_{P_x^a}$ and $\Pi_{Q_y^b}$ is a projector, it is clear that the $\tilde{P}_x^a$ (resp. $\tilde{Q}_y^b $) are positive semidefinite. Furthermore, using $\sum_a \Pi_{P_x^a}\leq \Id$ (since the projectors are, by definition, orthogonal), $\sum_i p_i=1$ and $ \Pileq{i} \leq \Id$ for every $i$ we get $\sum_a  \tilde{P}_x^a \leq \Id$ and hence $\tilde{P}_x^{garbage}\geq 0$ as well. 

Next we evaluate the strategy's success probability in the game. From the definition, 
\begin{align}
 \langle v_{\phi} | \tilde{P}_x^a \tilde{Q}_y^{b} | v_{\phi} \rangle &=    \langle v_{\phi} |  \Big(  \sum_i p_i  \Pileq{i} \Pi_{P_x^a} \Pileq{i}  \Big) \Big( \sum_j q_j  \Pileq{j} \Pi_{Q_y^b}\Pileq{j} \Big) | v_{\phi} \rangle \notag\\
& =   \sum_i p_i   \sum_j q_j \langle v_{\phi} |  \Pileq{i} \Pi_{P_x^a} \Pileq{i}  \Pileq{j} \Pi_{Q_y^b}\Pileq{j}  | v_{\phi} \rangle\notag\\
 &= \sum_i p_i   \sum_j q_j \langle v_{\phi} |  \Pi_{P_x^a} \Pileq{i}  \Pileq{j} \Pi_{Q_y^b} | v_{\phi} \rangle\notag \\
& = \sum_i p_i   \sum_j q_j \langle v_{P_x^a} |   \Pileq{i}  \Pileq{j} | v_{Q_y^b} \rangle \notag\\
&=  \langle v_{P_x^a} | v_{Q_y^b} \rangle.  \label{restricttoshortident}
\end{align}
Here the third equality follows since $| v_{\phi} \rangle$ is in the image of $\Pileq{i}$ for all $i$, the fourth equality follows from observation \ref{projectoraction} and the last from the definition of $ \Pileq{i}$ (using $i,j\geq 1$ wlog since $p_0=q_0=0$). Furthermore, 
\begin{align*}
\langle v_{\phi} | \tilde{P}_x^{garbage} \tilde{Q}_y^b | v_{\phi} \rangle &=    \langle v_{\phi} |  \Big ( I -\sum_{a \in \mathcal{A}} \tilde{P}_x^a \Big )\tilde{Q}_y^b  | v_{\phi} \rangle \\
&=   \langle v_{\phi} |\tilde{Q}_y^b  | v_{\phi} \rangle - \sum_{a \in \mathcal{A}}  \langle v_{\phi} |\tilde{P}_x^a  \tilde{Q}_y^b  | v_{\phi} \rangle \\
& = \langle v_{\phi}   | v_{Q_y^b} \rangle - \sum_{a \in \mathcal{A}} \langle v_{P_x^a} | v_{Q_y^b} \rangle \\
&= 0 ,
\end{align*}
where the fourth equality follows by equation \eqref{restricttoshortident}, and the third equality follows by reasoning very similar to that used to prove equation \eqref{restricttoshortident}.  Using similar arguments one can verify that $\langle v_{\phi} | \tilde{P}_x^a \tilde{Q}_y^{garbage} | v_{\phi} \rangle =   0 $ and $ \langle v_{\phi} | \tilde{P}_x^{garbage} \tilde{Q}_y^{garbage} | v_{\phi} \rangle =   0$ as well. Hence the ``garbage'' outcomes have probability zero of occurring (given the shared state is $\rho = \ket{v_{\phi}}\bra{v_{\phi}}$) and we may ignore them. 
       \end{proof}

\subsection{Commutator Bound}

\begin{theorem} \label{commutatorbound-nongarb}
Suppose the $p_i,q_j$ are such that 
$$\max\Big(\sum_i p_i^2,\,\sum_j q_j^2,\,\sum_i p_iq_i\Big)\,=\,O\Big(\frac{1}{N}\Big),$$
for instance $p_i = q_j = 1/(N-1)$ for $0< i,j < N$).  Then, for each value of $a,b,x,y \in A \times B \times X \times Y$, we have that  
$$\left\|\left [  \tilde{P}_x^a , \tilde{Q}^b_y \right ]  \right\|  =  O \left (  \frac{1}{\sqrt{N}}\right ).$$
\end{theorem}

\begin{proof}
Fix $(a,b,x,y) \in A \times B \times X \times Y$.  In order to simplify notation within this proof we write $\tildPiPxa$ for $\tilde{P}_x^a$, $\PiPxa$ for $\Pi_{P_x^a}$, $\tildPiQyb$ for $\tilde{Q}^b_y$, and $\PiQyb$ for $\Pi_{Q_y^b}$.
For any $1\leq i\leq N$ let $\Pi_{=i}=(\Id-\Pi_{<i})\Pi_{\leq i}(\Id-\Pi_{<i})$. Using that $\Pi_{<i}\leq \Pi_{\leq i}$ for each $i$, we get that the $\Pi_{=i}$ are orthogonal projectors and, taking the convention that $\Pi_{\leq 0}=0$, $\sum_{i\leq N} \Pi_{=i} = \Pi_{\leq N}$. 

Proposition~\ref{projectoroneshift} immediately implies that for any $i<N$ and $k>i+1$ it holds that $\Pi_{=k} P \pili=0$. Thus $\pili P\pili = (\Id - \pieip)P\pili$ and similarly for $Q$. Expand
\begin{align}
[\tilde{P},\tilde{Q}] &= \sum_{i,j} p_i q_j\, \big[\pili P \pili, \pilj Q \pilj\big]\notag\\
&= \sum_{i,j} p_iq_j \,\big( \pili P (\Id-\pieip) (\Id-\piejp) Q \pilj - \pilj Q (\Id-\piejp)(\Id-\pieip)P\pili\big)\notag\\
&= \sum_{i,j} p_iq_j\, \Big( \pili [P ,Q]  \pilj  +  \big(\pili P \pieip \piejp Q \pilj- \pilj Q \piejp \pieip P\pili\big) \notag\\
&\qquad -  \big(\pili P (\pieip+\piejp) Q \pilj -  \pilj Q (\pieip+\piejp) P \pili \big)    \Big).\label{eq:cb-1}
\end{align}
The second equality above follows by using Propositions \ref{projectorlevelidentities}, and \ref{projectoroneshift}.  We bound each of the four terms in~\eqref{eq:cb-1} separately. Using $i,j<N$ terms of the form $\pili [P ,Q]  \pilj$ evaluate to zero by Proposition~\ref{commutation-cancellation identity}. The second term 
\begin{align*}
\sum_{i,j} p_iq_j\,   \big(\pili P \pieip \piejp Q\pilj - \pilj Q \piejp \pieip P\pili\big) &= \sum_{i} p_iq_i\,\pili \big( P \pieip Q -  Q \pieip P\big)\pili,
\end{align*}
which using $\|  P \pieip Q -  Q \pieip P\|\leq 2$ and $\sum_i p_iq_i = O(N^{-1})$ has norm $O(N^{-1})$. It remains to bound the last two terms in~\eqref{eq:cb-1}. Towards this we first claim that
\begin{equation}
\Big\|\sum_i \pili P \pieip \Big\|= O\Big(\frac{1}{\sqrt{N}}\Big).\label{eq:cbound-2}
\end{equation}
 This can be seen by evaluating
\begin{align*}
\Big(\sum_i p_i \pili P \pieip\Big)\Big(\sum_i p_i \pili P \pieip\Big)^\dagger &= \sum_i p_i^2 \pili P \pieip P \pili\\
&\leq \sum_i p_i^2 \pili P \pili\\
&\leq \sum_i p_i^2 \Id,
\end{align*}
from which the bound~\eqref{eq:cbound-2} follows since $\sum_i p_i^2 = O(1/N)$. Together with the fact that $\|\sum_i p_i P\pili\|\leq 1$, and using analogous bounds for $Q$, the last two terms in~\eqref{eq:cb-1} each has norm at most $O(N^{-1/2})$. This concludes the proof. 
\end{proof}

\begin{corollary} \label{commutatorbound-garbcase}
Let us specify $p_i = q_j =  \frac{1}{N-1}$ when $0< i,j < N$  (and $p_0 = q_0 = p_N = q_N = 0$).  Then, for each value of $a,b,x,y \in A \times B \times X \times Y$, we have that

\begin{align*}
&\left \| \left [  \tilde{P}_x^{garbage} , \tilde{Q}_y^b \right ]  \right \|_2  =  O \left (   \frac{|A|}{\sqrt{N}}\right ), \\
&\left \| \left [  \tilde{P}_x^a , \tilde{Q}_y^{garbage} \right ]  \right \|_2  =  O \left (   \frac{|B|}{\sqrt{N}}\right ),\\
& \left \| \left [  \tilde{P}_x^{garbage} , \tilde{Q}_y^{garbage} \right ]  \right \|_2  =  O \left (   \frac{|A| |B|}{\sqrt{N}}\right ),
\end{align*}
\end{corollary}

\begin{proof}
\begin{align*}
\big [  \tilde{P}_x^{garbage} , \tilde{Q}_y^b \big ]  &=  \Big [ \Big ( I - \sum_{a \in \mathcal{A}} \tilde{P}_x^a \Big )  , \tilde{Q}_y^b \Big ]= 
\big [  I   , \tilde{Q}_y^b \big ] - \sum_{a \in \mathcal{A}} \big [  \tilde{P}_x^a, \tilde{Q}_y^b \big ] = 0 - \sum_{a \in \mathcal{A}} \big [  \tilde{P}_x^a , \tilde{Q}_y^b \big ], \\
 \big [  \tilde{P}_x^a , \tilde{Q}_y^{garbage} \big ] & = \Big [  \tilde{P}_x^a , \Big ( I - \sum_{b \in \mathcal{B}} \tilde{Q}_y^b \Big ) \Big ] =  -\big [  \tilde{P}_x^a , I \big ] +  \sum_{b \in \mathcal{B}} \big [  \tilde{P}_x^a , \tilde{Q}_y^b \big ] = 0 + \sum_{b \in \mathcal{B}} \big [  \tilde{P}_x^a , \tilde{Q}_y^b \big ],\\
\big [  \tilde{P}_x^{garbage} , \tilde{Q}_y^{garbage} \big ] &= \Big [ \Big ( I - \sum_{a \in \mathcal{A}} \tilde{P}_x^a \Big )  ,  \Big ( I - \sum_{b \in \mathcal{B}} \tilde{Q}_y^b \Big ) \Big ] \\
& = \Big [  I  ,  \Big ( I - \sum_{b \in \mathcal{B}} \tilde{Q}_y^b \Big ) \Big ]  - \sum_{a \in \mathcal{A}} \Big [  \tilde{P}_x^a , \Big ( I - \sum_{b \in \mathcal{B}} \tilde{Q}_y^b \Big )  \Big ] \\
&= 0 - \sum_{a \in \mathcal{A}} \Big ( -  \big [  \tilde{P}_x^a , I  \big ] +  \sum_{b \in \mathcal{B}}  \big [  \tilde{P}_x^a ,  \tilde{Q}_y^b  \big ] \Big ) \\ 
& =  \sum_{a \in \mathcal{A}}  \sum_{b \in \mathcal{B}}  \big [  \tilde{P}_x^a ,  \tilde{Q}_y^b  \big ],
\end{align*}

Using the triangle inequality and Theorem \ref{commutatorbound-nongarb} then gives
\begin{align*}
\left \| \left [  \tilde{P}_x^{garbage} , \tilde{Q}_y^b \right ]   \right \| & = \Big \|- \sum_{a \in \mathcal{A}} \left [  \tilde{P}_x^a , \tilde{Q}_y^b \right ]  \Big \| \leq  \sum_{a \in \mathcal{A}}  \left \| \left [  \tilde{P}_x^a , \tilde{Q}_y^b \right ]  \right \| \leq |A| O \left (  \frac{1}{\sqrt{N}}\right ) =  O \left (  \frac{|A|}{\sqrt{N}}\right ), \\
\left \|\left [  \tilde{P}_x^a , \tilde{Q}_y^{garbage} \right ]   \right \| & = \Big \|  \sum_{b \in \mathcal{B}} \left [  \tilde{P}_x^a , \tilde{Q}_y^b \right ]  \Big \| \leq  \sum_{b \in \mathcal{B}} \left \|  \left [  \tilde{P}_x^a , \tilde{Q}_y^b \right ]  \right \|  \leq |B| O \left (  \frac{1}{\sqrt{N}}\right ) = O \left (  \frac{|B|}{\sqrt{N}}\right ), \\
\left \| \left [  \tilde{P}_x^{garbage} , \tilde{Q}_y^{garbage} \right ]  \right \|  & = \Big \| \sum_{a \in \mathcal{A}}  \sum_{b \in \mathcal{B}}  \left [  \tilde{P}_x^a ,  \tilde{Q}_y^b  \right ]    \Big \|  \leq   \sum_{a \in \mathcal{A}}  \sum_{b \in \mathcal{B}} \left \|   \left [  \tilde{P}_x^a ,  \tilde{Q}_y^b  \right ]    \right \|  \\
&\leq |A| |B| O \left (  \frac{1}{\sqrt{N}}\right ) = O \left (  \frac{|A||B|}{\sqrt{N}}\right ),
\end{align*}
This is the desired result.
\end{proof}

\section{A lower bound on $\MIP_\delta^*$}\label{sec:nexp}

In this section we give a detailed sketch of the proof of Theorem~\ref{thm:nexp}. 
The proof closely follows the lines of Theorem~2 in~\cite{ikm09}, where the same result is proved for the case of provers that are restricted to be perfectly commuting. Most of the work consists in carefully going through the argument in~\cite{ikm09} and verifying that the commutation condition, provided it is satisfied for a sufficiently small $\delta$, suffices to preserve soundness. Although intuitively one expects this to be the case, one still has to be a little careful in order to avoid any dimension dependence coming into the argument.

\subsection{Proof outline}

Our starting point is a non-adaptive $3$-query PCP for $\NEXP$ with perfect completeness and soundness bounded away from $1$ and in which the alphabet size is a single bit. Fix an input $x$, let $N$ be the length of the PCP and $\pi:[N]^3\to [0,1]$ the distribution on queries. We may assume that the marginal distribution of $\pi$ on any of its three coordinates is uniform. Let $V: [N]^3 \times \{0,1\}^3\to \{0,1\}$ be the acceptance predicate. We consider the following protocol in which the verifier interacts with two provers only:  

\begin{enumerate}
\item The verifier chooses a triple $(i_1,i_2,i_3)$ according to $\pi$, $i\in\{i_1,i_2,i_3\}$ uniformly at random and $j\in [N]$ uniformly at random. He sends (lexicographically ordered) tuples $\{i_1,i_2,i_3\}$ to the first prover and $\{i,j\}$ to the second. 
\item $A$ replies with three bits $a_{i_1},a_{i_2},a_{i_3}$. $B$ replies with two bits $b_i,b_j$. 
\item The referee accepts if and only if $V(i_1,i_2,i_3,a_{i_1},a_{i_2},a_{i_3})=1$ and $a_{i}=b_i$. 
\end{enumerate}

This protocol is obtained through the standard oracularization technique, except for the additional question $j$ sent to the second prover. This is called a ``dummy question'' in~\cite{ikm09} and it plays an important role in the analysis. 

First we note that the completeness property of the protocol trivially holds. Hence it remains to establish soundness. This is done in the following lemma, which is proved in the following section: 

\begin{lemma}\label{lem:nexp-soundness}
There exists a universal constant $c_1>0$ such that the following holds. Let $(A_{i_1i_2i_3}^{a_1a_2a_3},B_{ij}^{b_ib_j},\rho)$ be a $2$-prover $\delta$-AC projective strategy that succeeds with probability at least $1-\eps$ in the protocol. Then there exists an assignment to the variables $[N]$ that satisfies the PCP verifier's queries with probability at least $1-O(N^2(\eps^{1/2}+\delta^{c_1}))$. 
\end{lemma}

Theorem~\ref{thm:nexp} follows immediately from the above lemma (using Claim~\ref{claim:projective} to argue that the assumption that the provers' measurements are projective is without loss of generality).

\subsection{Soundness analysis}

In this section we sketch the proof of Lemma~\ref{lem:nexp-soundness}. Given a strategy for the provers for every $i\in [N]$ we define the following POVM with outcomes in $\{0,1\}$: 
\beq\label{eq:def-c}
 C_i^c \,:=\, \Es{j\in [N]} \sum_{c'\in\{0,1\}} B_{ij}^{cc'},
\eeq
where the expectation is taken with respect to the uniform distribution. Define a (probabilistic) assignment $(c_i)_{i\in [N]}$ to the PCP proof according to the distribution 
\beq\label{eq:def-ass}
\text{Prob}\big((c_i)\big) \,:=\, \Tr\big( \sqrt{{C}_{N}^{c_N}}\cdots \sqrt{{C}_{1}^{c_1}} \rho \sqrt{{C}_{1}^{c_1}}\cdots  \sqrt{{C}_{N}^{c_N}}\big).
\eeq
We will show that this assignment satisfies the acceptance predicate with good probability. First we prove the following claim, which gives a simpler form for the marginals of the distribution on assignments to any three fixed variables.

\begin{claim}\label{claim:sat}
There exists a constant $c_1>0$ such that the following holds for any $i,j,k\in [N]$:
\begin{align}
\sum_{a,b,c}\Big| \Tr\Big( &\sqrt{{C}_{i}^{c}}\sqrt{ C_{j}^{b}} \sqrt{C_{k}^{a}} \rho\sqrt{{C}_{i}^{a}}\sqrt{ C_{j}^{b}} \sqrt{C_{k}^{c}}  \Big)\notag\\
&- \sum_{\substack{c_\ell\\ \ell\notin\{i,j,k\}}}  \Tr\Big( \sqrt{{C}_{N}^{c_N}}\cdots \sqrt{{C}_{1}^{c_1}} \rho \sqrt{{C}_{1}^{c_1}}\cdots  \sqrt{{C}_{N}^{c_N}}\Big) \Big|\,=\, O\big(N^2\big(\eps^{1/2}+\delta^{c_1}\big)\big)\label{eq:sat-1}
\end{align}
\end{claim}

\begin{proof}[Proof sketch]
First we note that the bound~\eqref{eq:sat-1} follows by an easy induction once the following has been established: for any $t$ and $j_1,\ldots,j_t\in [N]$, 
\begin{align}\label{eq:sat-2}
\Es{i}\,\sum_{c_1,\ldots,c_t}\,\Big| \sum_a \,&\Tr\Big( \sqrt{{C}_{j_t}^{c_{t}}}\cdots\sqrt{ {C}_{j_1}^{c_{1}}}\sqrt{{C}_{i}^{a}} \rho\sqrt{{C}_{i}^{a}} \sqrt{{C}_{j_1}^{c_{1}}} \cdots \sqrt{{C}_{j_t}^{c_{t}}} \Big)\notag\\
&- \Tr\Big( \sqrt{{C}_{j_t}^{c_{t}}}\cdots\sqrt{ {C}_{j_1}^{c_{1}}}\rho\sqrt{{C}_{j_1}^{c_{1}}} \cdots \sqrt{{C}_{j_t}^{c_{t}}} \Big) \Big|\,=\, O\big(t\,\big(\eps^{1/2} + \delta^{1/16}\big)\big).
\end{align}
To prove~\eqref{eq:sat-2}, for any $i\in [N]$ and $a\in \{0,1\}$ we introduce the POVM element 
$$\hat{A}_i^a := \Es{j,k} \sum_{a_j,a_k} A_{ijk}^{a a_ja_k},$$
where the expectation is taken according to the conditional distribution $\pi(\cdot,\cdot|i)$. Note that $\sum_a \hat{A}_i^a = \Id$, and success in the consistency check of the protocol implies
\beq\label{eq:confuse-com-1}
\Es{i} \sum_a \Tr\big( \hat{A}_i^a \sqrt{C_i^a} \rho \sqrt{C_i^a} \big) \geq 1-\eps.
\eeq
This justifies applying Claim~\ref{claim:sq-bound}, and from~\eqref{eq:sq-bound-trace} we get 
\beq\label{eq:sat-3}
\Es{i} \, \Big\| \sum_a\, \sqrt{C_i^a}\rho \sqrt{C_i^a}-\sqrt{\hat{A}_i^a}\rho\sqrt{\hat{A}_i^a}\Big\|_1 \,=\, O\big(\eps^{1/2} + \delta^{1/16}\big).
\eeq
Using~\eqref{eq:sat-3} we obtain 
\begin{align}
\Es{i}\sum_{c_1,\ldots,c_t}\Big| \sum_a& \Tr\big( \sqrt{{C}_{j_t}^{c_{t}}}\cdots\sqrt{ {C}_{j_1}^{c_{1}}}\sqrt{{C}_{i}^{a}} \rho\sqrt{{C}_{i}^{a}} \sqrt{{C}_{j_1}^{c_{1}}} \cdots \sqrt{{C}_{j_t}^{c_{t}}} \big)\notag \\
&-  \sum_a \Tr\big( \sqrt{{C}_{j_t}^{c_{t}}}\cdots\sqrt{ {C}_{j_1}^{c_{1}}}\sqrt{\hat{A}_{i}^{a}} \rho\sqrt{\hat{A}_{i}^{a}} \sqrt{{C}_{j_1}^{c_{1}}} \cdots \sqrt{\hat{C}_{j_t}^{c_{t}}} \big) \Big|\,=\, O\big(\eps^{1/2} + \delta^{1/16}\big).\label{eq:sat-4}
\end{align}
Applying the $(4\delta)$-AC condition (together with~\eqref{eq:com-bound} in order to apply it to the square roots) between $\hat{A}$ and $C$ $(2t)$ times leads to~\eqref{eq:sat-2}. To conclude the proof of~\eqref{eq:sat-1} we start from the second term in the absolute value and apply~\eqref{eq:sat-2} $(N-3)$ times to eliminate the $C$ that are being summed over. 
\end{proof}

Our second claim relates the marginal computed in Claim~\ref{claim:sat} to the original provers' strategy. 

\begin{claim}\label{claim:sat2}
There exists a constant $c_1>0$ such that the following holds:
\beq\label{eq:sat2-1}
\Es{ijk}\sum_{a,b,c}\Big| \Tr\Big( \sqrt{{C}_{i}^{c}}\sqrt{ C_{j}^{b}} \sqrt{C_{k}^{a}} \rho\sqrt{{C}_{i}^{a}}\sqrt{ C_{j}^{b}} \sqrt{C_{k}^{c}}  \Big)- \Tr\big( A_{ijk}^{abc}\rho\big) \Big|\,=\, O\big(\eps^{1/2}+\delta^{1/16}\big),
\eeq
where the expectation is taken according to the distribution $\pi$ used in the protocol. 
\end{claim}

\begin{proof}[Sketch]
We first proceed as in the proof of Claim~\ref{claim:sat} and note that success in the consistency check of the protocol implies, through Claim~\ref{claim:sq-bound}, the bound
\beq\label{eq:sat2-2}
\Es{ijk} \, \Big\| \sum_a\, \sqrt{C_i^a}\rho \sqrt{C_i^a}-\sqrt{{A}_{ijk}^a}\rho\sqrt{A_{ijk}^a}\Big\|_1 \,=\, O\big(\eps^{1/2} + \delta^{1/16}\big).
\eeq
Note however the slightly different formulation from~\eqref{eq:sat-3}, where we left the expectation on $j$ and $k$ outside, and slightly abused notation to write $A_{ijk}^a$ for $\sum_{b,c} A_{ijk}^{abc}$. Applying~\eqref{eq:sat2-2} thrice, using the $(4\delta)$-AC condition between $A$ and $C$ and the fact that we assumed the $A_{ijk}$ to be projective measurements, we get 
\begin{align*}
\Es{ijk}\Big\|\sum_{a,b,c} \,\sqrt{{C}_{i}^{c}}\sqrt{ C_{j}^{b}} \sqrt{C_{k}^{a}} \rho\sqrt{{C}_{i}^{a}}\sqrt{ C_{j}^{b}} \sqrt{C_{k}^{c}}  - \sum_{a,b,c}\, A_{ijk}^{abc} \rho A_{ijk}^{abc}\Big\|_1\,=\, O\big(\eps^{1/2} + \delta^{1/16}\big).
\end{align*}
The claimed bound~\eqref{eq:sat2-1} follows by noting that the $A_{ijk}^{abc}$ are orthogonal for different outputs, and using the following pinching inequality: for any $X$ and projection $P$, $\|PXP+(1-P)X(1-P)\|_1 \leq \|X\|_1$. 
\end{proof}

Let $(c_i)_{i\in [N]}$ be sampled according to the distribution~\eqref{eq:def-ass}. Combining the bounds from Claim~\ref{claim:sat} and Claim~\ref{claim:sat2} we see that for any query $(i,j,k)$ made by the PCP verifier the marginal distribution on assignments induced by $(c_i)$ is within statistical distance $O(N^2(\eps^{1/2}+\delta^{c_1}))$ from the distribution on assignments obtained from the entangled provers' answers in the protocol. Since by assumption the latter satisfies the PCP predicate with probability $1-\eps$, we deduce that the assignment $(c_i)_{i\in [N]}$ satisfies a randomly chosen query of the PCP verifier with probability (over the sampling of $(c_i)$ as well as over the PCP verifier's random choice of query) at least $1- O(N^2(\eps^{1/2}+\delta^{c_1}))$. This completes the proof of Lemma~\ref{lem:nexp-soundness}, from which Theorem~\ref{thm:nexp} follows easily. 

\section{Discussion}\label{sec:discussion}

The rounding scheme for the QC SDP hierarchy in Section \ref{commutatorbound}, and our introduction of the corresponding class $\AMIP^*$ in Section \ref{sec:nexp}, are motivated by a desire to develop a framework for the study of quantum multiprover interactive proof systems that is both computationally bounded and relevant for typical applications of such proof systems. Our main technical result, Theorem~\ref{thm:sdp}, demonstrates the first aspect. In this section we discuss the relevance of the new model, its connection with the standard definition of $\MIP^*$, and applications to quantum information. 
   
%Independently of its relationship to $\MIP^*$, the class $\AMIP^*$ already finds applications to quantum information and in particular to device-independent cryptography. In Section~\ref{subsec::SPM} we discuss an application to the study of Device Independent Randomness Generation for devices with weak cross-talk~\cite{SPM13}.

\subsection{Commuting approximants: some results, limits, and possibilities}\label{sec:commute}

While we believe $\AMIP^*$ is of interest in itself, we do not claim that approximately commuting provers are more natural than commuting provers, or provers in tensor product form; the main goal in introducing the new class is to shed light on its thus-far-intractable parent $\MIP^*$. In light of the results from Section~\ref{commutatorbound} the relationship between the two classes seems to hinge on the general mathematical problem of finding exactly commuting approximants to approximately commuting matrices.    

\subsubsection{Limits for commuting approximants}\label{sec:voiculescu}

The main objection to the existence of a positive answer for the ``commuting approximants" question is revealed by a beautiful construction of Voiculescu who exhibits a surprisingly simple scenario in which commuting approximants provably do not exist~\cite{Voi83}. The following  is a direct consequence of Voiculescu's result.

\begin{theorem}[Voiculescu]\label{thm::Voiculescu}
For every $d \in \mathbb{N}$ there exists a pair of unitary matrices $U_1, U_2 \in \mathbb{C}^{d \times d}$ with $\| [ U_1, U_2 ] \| = O(\frac{1}{d})$, such that for any pair of complex matrices $A, B \in \mathbb{C}^{d \times d}$ satisfying $[A, B] = 0$, $\max ( \| U_1 - A \|, \| U_2 - B \| ) = \Omega(1)$.  
\end{theorem}

In Voiculescu's example $U_1$ is a $d$-dimensional cyclic permutation matrix, and $U_2$ is a diagonal matrix whose eigenvalues are the $d^{th}$ roots of unity.  The proof draws on a connection to homology, in particular using a homotopy invariant to establish the lower bound on the distance to commuting approximants.  A succinct and elementary proof of the result is given by Exel and Loring \cite{EL89}.  
  
In the context of entangled strategies one is most concerned with Hermitian matrices representing measurements, rather than unitaries.  However, as a consequence of Theorem \ref{thm::Voiculescu} we see that if one considers the Hermitian operators $M^j_k = \frac{(-i)^j}{2} (U_k + (-1)^j U_k^{\dagger})$ ($j \in \{0,1\}$) we have that $\| [M^j_1, M^{j'}_2] \| = O(\frac{1}{d})$, and yet any exactly commuting set of matrices must be a constant distance away in the operator norm.  Thus Theorem \ref{thm::Voiculescu} rules out the strongest form of a ``commuting approximants" statement, which would ask for approximants in the same space as the original matrices, and with a commutator bound that does not depend on the dimension of the matrices.  

Thus Theorem~\ref{thm::Voiculescu} invites us to refine the ``commuting approximants'' question and distinguish ways in which it may avoid the counter-example; we describe some possibilities in the following subsections. 

\subsubsection{Ozawa's conjecture}

Motivated by the study of Tsirelson's problem~\cite{ScholzW08} and the relationship with Tsirelson's conjecture, Ozawa~\cite{Ozawa13commuting} introduces two equivalent conjectures, the ``Strong Kirchberg Conjecture (I)'' and ``Strong Kirchberg Conjecture (II)'' respectively, which postulate the existence of commuting approximants to approximately commuting sets of POVM measurements and unitaries respectively. The novelty of these conjectures, which allows them to avoid the immediate pitfall given by Voiculescu's example, is that Ozawa considers approximants in a larger Hilbert space than the original approximately commuting operators. Precisely, his Strong Kirchberg Conjecture (I) states the following:

\begin{conjecture}[Ozawa]\label{conj:ozawa}
Let $m,\ell\geq 2$ be such that $(m,\ell)\neq (2,2)$ \footnote{The case $(m,\ell) = (2,2)$ is the only nontrivial setting for which we have some understanding. In particular nonlocal games with two inputs and two outputs per party can be analyzed via an application of Jordan's lemma~\cite{Masanes05}.}. For every $\kappa>0$ there exists $\eps>0$ such that, if $\dim\mathcal{H}<\infty$ and $(P_i^k)$, $(Q_j^l)$ is a pair of $m$ projective $\ell$-outcome POVMs on $\mathcal{H}$ satisfying $\|[P_i^k,Q_j^l]\|\leq \eps$, then there is a finite-dimensional Hilbert space $\tilde{H}$ containing $\mathcal{H}$ and projective POVMs $\tilde{P}_i^k$, $\tilde{Q}_j^l$ on $\tilde{\mathcal{H}}$ such that $\|[\tilde{P}_i^k,\tilde{Q}_j^l]\|=0$ and $\|\Phi_\mathcal{H}(\tilde{P}_i^k)-P_i^k\|\leq\kappa$ and $\|\Phi_\mathcal{H}(\tilde{Q}_j^l)-Q_j^l\|\leq\kappa$.  Here $\Phi_\mathcal{H}$ denotes the compression to $\mathcal{H}$, defined by $\Phi_{\mathcal{H}} (M) \equiv P_{\mathcal{H}} M P_{\mathcal{H}}$, where $P_{\mathcal{H}}$ is the projection onto $\mathcal{H}$. 
\end{conjecture}

Ozawa gives an elegant proof of a variant of the conjecture that applies to just two approximately commuting unitaries, thereby establishing that extending the Hilbert space can allow one to avoid the complications in Voiculescu's example. He also establishes that the conjecture is \emph{stronger} than Kirchberg's conjecture (itself equivalent to Tsirelson's problem and Connes' embedding conjecture~\cite{JungeNPPSW11}), casting doubt, if not on its validity, at least on its approachability. 

Nevertheless, we can mention the following facts. First, it follows from Theorem~\ref{commutatorbound-nongarb} that Conjecture~\ref{conj:ozawa} implies the equality $\AMIP^* = \MIP^*$; in fact it implies that $\MIP^*_\delta = \MIP^*$ for small enough $\delta$, depending on how the parameter $\eps$ in Conjecture~\ref{conj:ozawa} depends on $\kappa$, $m$ and $d$. For this it suffices to verify that a state $\rho$ optimal for a strategy based on POVMs $P_i^k$ and $Q_j^l$ in a given protocol can be lifted to a state $\tilde{\rho}$ on $\tilde{\mathcal{H}}$ such that the correlations exhibited by performing the POVMs $\tilde{P}_i^k$, $\tilde{Q}_j^l$ on $\tilde{\rho}$ approximately reproduce those generated by $P_i^k$, $Q_j^l$ on $\rho$; this is easily seen to be the case provided $\kappa$ is small enough.  Therefore, Theorem~\ref{commutatorbound-nongarb} may be seen as an elementary proof that Conjecture~\ref{conj:ozawa} implies that $\MIP^*$ is computable (this result was previously known, but previous methods use a series of reductions connecting Kirchberg's Conjecture to Tsirelson's problem, see \cite{Ozawa13commuting}).   

Second, Conjecture~\ref{conj:ozawa} can be weakened in several ways without losing the implication that $\AMIP^* = \MIP^*$.  For instance, it is not necessary for the exactly commuting $\tilde{P}_i^k$, $\tilde{Q}_j^l$ to approximate $P_i^k$, $Q_j^l$ in operator norm --- in our context of interactive proofs, only the correlations obtained by measuring a particular state need to be preserved, and this does not in general imply an approximation as strong as that promised in Conjecture~\ref{conj:ozawa}. 

%In the next subsection we discuss another weakening of Conjecture~\ref{conj:ozawa} that our results invite us to consider, and consists of allowing a (weak) dependence of the parameter $\eps$ on the dimension of the space $\mathcal{H}$. 

\subsubsection{Dimension dependent bounds}

An alternative relaxation for the ``commuting approximants'' question is to allow the approximation error to depend explicitly on the dimension of the matrices. A careful analysis of the rounding scheme from Theorem \ref{thm:sdp} shows that it produces $d$-dimensional POVM elements with an $O(1/\sqrt{\log(d)})$ bound on the commutators  (this is because the dimension of the subspace $\Ima( \Pileq{N-1})$ is exponential in $N$). Unfortunately, Voiculescu's result (Theorem \ref{thm::Voiculescu}) shows that one can only hope for good approximants in the operator norm if the commutator bound is $o(1/d)$. It remains instructive to find \emph{any} explicit existence result for commuting approximants in the general case, regardless of dimension dependence. Concretely, we conjecture that Conjecture~\ref{conj:ozawa} may be true with a parameter $\kappa$ that scales with the dimension $d$ of the operators $\{P_i^k,Q_j^l\}$ as $\kappa = \eps^c \poly(d)^{(ml)^2}$ for some constant $0<c\leq 1$. 

\subsubsection{An alternative norm}

Another relaxation of the ``commuting approximants" question, which would be sufficient to imply $\AMIP^* = \MIP^*$, is to allow for any set of commuting approximants which approximately preserves the winning probability of the game.  
%If one wishes to find approximants of the same dimension, and with dimension-independent error bound, then this relaxation has the advantage of allowing for a constant sized error in the operator norm, which is certainly necessary as evidenced by Voiculescu's example (Theorem \ref{thm::Voiculescu}).  More generally, if one allows for an extended Hilbert space and dimension-dependent error term, this relaxation has the advantage of being strictly weaker than Conjecture \ref{conj:ozawa}, while still implying $\AMIP^* = \MIP^*$.  
For concreteness we include a precise version of a possible statement along these lines:

\begin{conjecture}
There exists a function $f(\eps, k):  \mathbb{R}^+ \times \mathbb{N} \to \mathbb{R}^{+}$ satisfying $\lim_{\eps \to 0}f(\eps, k) = 0$ for all $k \in \mathbb{N}$, such that for every game $G$ and $(m,\ell)$ strategy $(A_x^{a}, B_y^b,\rho)$ which is $\delta$-AC, there exists a $0$-AC strategy $(\tilde{A}_x^a , \tilde{B}_{y}^{b}, \rho)$ for $G$ satisfying
$$\left |  \omega^*\big( ((A_{x}^{a}, B_{y}^{b},\rho);G\big) -  \omega^*\big( (\tilde{A}_x^a , \tilde{B}_{y}^{b}, \rho); G \big)  \right|  \leq f(\delta,m\ell).$$
\end{conjecture}

\subsection{Device-independent randomness expansion and weak cross-talk }\label{subsec::SPM}

A device-independent randomness expansion (DIRE) protocol is a protocol which may be used by a classical verifier to certify that a pair of untrusted devices are producing true randomness. Under the sole assumptions that the devices do not communicate with each other, and that the verifier has access to a small initial seed of uniform randomness, the protocol allows for the generation of much larger quantities of certifiably uniform random bits; hence the term ``randomness expansion''.  This conclusion relies only on the assumption that the two devices do not communicate, and in particular does not require any limit on the computational power of the devices, as is typically the case in the study of pseudorandomness.  The precise formalization of DIRE protocols is rather involved, and we direct the interested reader to the flourishing collection of works on the topic~\cite{Colbeck2011,Pironio2010,Vazirani2012,Fehr2013,acin2012randomness,um2013experimental,gallego2013device,ms14,csw14}.  In particular the precise formulation of the model is a focus of \cite{CVY13}.

Our definition of $\AMIP^*$ is directly relevant to the notion of devices with \emph{weak cross-talk} introduced in \cite{SPM13} as a model which relaxes the assumption that the devices must not communicate, leading to protocols that are more robust to leakage than the traditional model of device-independence.~\cite{SPM13} proposes the use of the QC SDP hierarchy in order to optimize over the set of ``weakly interacting'' quantum strategies that they introduce, but no bounds are shown on the rate of convergence.  This is where $\AMIP^*$ becomes relevant.  Our notion of $\delta$-AC strategies is easily seen to be a relaxation of weak cross-talk, and thus the analogue of the approach in \cite{SPM13} when performed with a $\delta$-AC constraint is at least as robust as the weak cross-talk approach.  Our rounding scheme for the QC SDP hierarchy thus provides a specific algorithm and complexity bound that applies to both  $\delta$-AC strategies and strategies with weak cross-talk.

\bibliography{commute}

\newcommand{\etalchar}[1]{$^{#1}$}
\begin{thebibliography}{BOGKW88}
\expandafter\ifx\csname urlstyle\endcsname\relax
  \providecommand{\doi}[1]{doi:\discretionary{}{}{}#1}\else
  \providecommand{\doi}{doi:\discretionary{}{}{}\begingroup
  \urlstyle{rm}\Url}\fi

\bibitem[ALM{\etalchar{+}}98]{AroLunMotSudSze98JACM}
S.~Arora, C.~Lund, R.~Motwani, M.~Sudan, and M.~Szegedy.
\newblock Proof verification and the hardness of approximation problems.
\newblock \emph{J. ACM}, 45(3):501--555, 1998.

\bibitem[AMP12]{acin2012randomness}
A.~Ac{\'\i}n, S.~Massar, and S.~Pironio.
\newblock Randomness versus nonlocality and entanglement.
\newblock \emph{Physical review letters}, 108(10):100402, 2012.

\bibitem[Ara02]{Arvind:02}
P.~K. Aravind.
\newblock The magic squares and {B}ell's theorem.
\newblock Technical report, ar{X}iv:quant-ph/0206070, 2002.

\bibitem[AS98]{AroSaf98JACM}
S.~Arora and S.~Safra.
\newblock Probabilistic checking of proofs: A new characterization of {NP}.
\newblock \emph{J. ACM}, 45(1):70--122, 1998.

\bibitem[Bel64]{Bell:64a}
J.~S. Bell.
\newblock On the {E}instein-{P}odolsky-{R}osen paradox.
\newblock \emph{Physics}, 1:195--200, 1964.

\bibitem[BFL91]{BabForLun91CC}
L.~Babai, L.~Fortnow, and C.~Lund.
\newblock Non-deterministic exponential time has two-prover interactive
  protocols.
\newblock \emph{Comput. Complexity}, 1:3--40, 1991.

\bibitem[Bha97]{bhatia}
R.~Bhatia.
\newblock \emph{Matrix Analysis}.
\newblock Number 169 in {Graduate Texts in Mathematics}. Springer-Verlag, {New
  York}, 1997.

\bibitem[BOGKW88]{BenGolKilWig88STOC}
M.~Ben-Or, S.~Goldwasser, J.~Kilian, and A.~Wigderson.
\newblock Multi-prover interactive proofs: How to remove intractability
  assumptions.
\newblock In \emph{Proceedings of the 20th Annual ACM Symposium on Theory of
  Computing (STOC)}, pages 113--131. 1988.

\bibitem[BSS14]{BancalSS14data}
J.-D. Bancal, L.~Sheridan, and V.~Scarani.
\newblock More randomness from the same data.
\newblock \emph{New Journal of Physics}, 16(3):033011, 2014.

\bibitem[CHSH69]{Clauser:69a}
J.~F. Clauser, M.~A. Horne, A.~Shimony, and R.~A. Holt.
\newblock {Proposed experiment to test local hidden-variable theories}.
\newblock \emph{Phys. Rev. Lett.}, 23:880--884, 1969.

\bibitem[CHTW04]{CHTW04}
R.~Cleve, P.~H{\o}yer, B.~Toner, and J.~Watrous.
\newblock Consequences and limits of nonlocal strategies.
\newblock In \emph{Proc. 19th IEEE Conf. on Computational Complexity (CCC'04)},
  pages 236--249. IEEE Computer Society, 2004.

\bibitem[CK11]{Colbeck2011}
R.~Colbeck and A.~Kent.
\newblock {Private randomness expansion with untrusted devices}.
\newblock \emph{Journal of Physics A: Mathematical and \ldots}, pages 1--11,
  2011.

\bibitem[Con76]{Connes76}
A.~Connes.
\newblock Classification of injective factors cases $ii_1$, $ii_\infty$,
  $iii_\lambda$, $\lambda\neq 1$.
\newblock \emph{Annals of Mathematics}, 104(1):pp. 73--115, 1976.

\bibitem[CSW14]{csw14}
K.-M. Chung, Y.~Shi, and X.~Wu.
\newblock Physical randomness extractors.
\newblock \emph{\textnormal{arXiv:1402.4797}}, 2014.

\bibitem[CVY13]{CVY13}
M.~Coudron, T.~Vidick, and H.~Yuen.
\newblock Robust randomness amplifiers: Upper and lower bounds.
\newblock In \emph{APPROX-RANDOM}, pages 468--483. 2013.

\bibitem[DLTW08]{DLTW08}
A.~C. Doherty, Y.-C. Liang, B.~Toner, and S.~Wehner.
\newblock The quantum moment problem and bounds on entangled multi-prover
  games.
\newblock In \emph{Proc. 23rd IEEE Conf. on Computational Complexity (CCC'08)},
  pages 199--210. 2008.

\bibitem[EL89]{EL89}
R.~Exel and T.~Loring.
\newblock Almost commuting unitary matrices.
\newblock \emph{Proceedings of the American Mathematical Society},
  106(4):913--915, 1989.

\bibitem[EPR35]{epr}
A.~Einstein, B.~Podolsky, and N.~Rosen.
\newblock Can quantum-mechanical description of physical reality be considered
  complete?
\newblock \emph{Physical Review}, 47:777--780, 1935.

\bibitem[FGL{\etalchar{+}}96]{FeiGolLovSafSze96JACM}
U.~Feige, S.~Goldwasser, L.~Lov\'asz, S.~Safra, and M.~Szegedy.
\newblock Interactive proofs and the hardness of approximating cliques.
\newblock \emph{J. ACM}, 43(2):268--292, 1996.

\bibitem[FGS13]{Fehr2013}
S.~Fehr, R.~Gelles, and C.~Schaffner.
\newblock {Security and composability of randomness expansion from Bell
  inequalities}.
\newblock \emph{Physical Review A}, pages 1--12, 2013.

\bibitem[GL{\etalchar{+}}13]{gallego2013device}
R.~Gallego~L{\'o}pez et~al.
\newblock Device-independent information protocols: measuring dimensionality,
  randomness and nonlocality.
\newblock 2013.

\bibitem[IKM09]{ikm09}
T.~Ito, H.~Kobayashi, and K.~Matsumoto.
\newblock Oracularization and two-prover one-round interactive proofs against
  nonlocal strategies.
\newblock In \emph{Proc. 24th IEEE Conf. on Computational Complexity (CCC'09)},
  pages 217--228. IEEE Computer Society, 2009.

\bibitem[IV12a]{IV12}
T.~Ito and T.~Vidick.
\newblock A multi-prover interactive proof for {NEXP} sound against entangled
  provers.
\newblock \emph{Proc. 53rd FOCS}, pages 243--252, 2012.

\bibitem[IV12b]{iv13arxiv}
T.~Ito and T.~Vidick.
\newblock A multi-prover interactive proof for {NEXP} sound against entangled
  provers.
\newblock Technical report, ar{X}iv:1207.0550, 2012.

\bibitem[JNP{\etalchar{+}}11]{JungeNPPSW11}
M.~Junge, M.~Navascues, C.~Palazuelos, D.~Perez-Garcia, V.~B. Scholz, and R.~F.
  Werner.
\newblock Connes' embedding problem and tsirelson's problem.
\newblock \emph{J. Math. Physics}, 52(1):012102, 2011.

\bibitem[Kir93]{Kirchberg93qwep}
E.~Kirchberg.
\newblock On non-semisplit extensions, tensor products and exactness of group
  ${C}^*$-algebras.
\newblock \emph{Inventiones mathematicae}, 112(1):449--489, 1993.

\bibitem[KM03]{KobMat03JCSS}
H.~Kobayashi and K.~Matsumoto.
\newblock Quantum multi-prover interactive proof systems with limited prior
  entanglement.
\newblock \emph{Journal of Computer and System Sciences}, 66(3):429--450, 2003.

\bibitem[KRT10]{KRT10}
J.~Kempe, O.~Regev, and B.~Toner.
\newblock Unique games with entangled provers are easy.
\newblock \emph{SIAM J. Comput.}, 39(7):3207--3229, 2010.

\bibitem[Lau03]{Laurent03lasserre}
M.~Laurent.
\newblock A comparison of the {S}herali-{A}dams, {L}ov{\'a}sz-{S}chrijver, and
  {L}asserre relaxations for 0-1 {P}rogramming.
\newblock \emph{Mathematics of Operations Research}, 28(3):470--496, 2003.

\bibitem[Mas05]{Masanes05}
L.~Masanes.
\newblock Extremal quantum correlations for n parties with two dichotomic
  observables per site.
\newblock Technical report, arXiv:quant-ph/0512100, 2005.

\bibitem[MS14]{ms14}
C.~A. Miller and Y.~Shi.
\newblock Robust protocols for securely expanding randomness and distributing
  keys using untrusted quantum devices.
\newblock In \emph{Proc. 46th STOC}. ACM, New York, NY, USA, 2014.

\bibitem[NPA07]{NPA07}
M.~Navascu{\'{e}}s, S.~Pironio, and A.~Ac{\'{\i}}n.
\newblock Bounding the set of quantum correlations.
\newblock \emph{Phys. Rev. Lett.}, 98:010401, Jan 2007.

\bibitem[NPA08a]{NPA08NJP}
M.~Navascu{\'{e}}s, S.~Pironio, and A.~Ac{\'{\i}}n.
\newblock A convergent hierarchy of semidefinite programs characterizing the
  set of quantum correlations.
\newblock \emph{New Journal of Physics}, 10(073013), 2008.

\bibitem[NPA08b]{pna08}
M.~Navascu{\'{e}}s, S.~Pironio, and A.~Ac{\'{\i}}n.
\newblock A convergent hierarchy of semidefinite programs characterizing the
  set of quantum correlations.
\newblock Technical report, {arXiv}:0803.4290v1 [quant-ph], 2008.

\bibitem[Oza13]{Ozawa13commuting}
N.~Ozawa.
\newblock Tsirelson's problem and asymptotically commuting unitary matrices.
\newblock \emph{Journal of Mathematical Physics}, 54(3):032202, 2013.

\bibitem[PAM10]{Pironio2010}
S.~Pironio, A.~Ac\'{\i}n, and S.~Massar.
\newblock {Random numbers certified by Bell's theorem}.
\newblock \emph{Nature}, pages 1--26, 2010.

\bibitem[PV10]{PalV10I3322}
K.~F. P\'al and T.~V\'ertesi.
\newblock Maximal violation of a bipartite three-setting, two-outcome {B}ell
  inequality using infinite-dimensional quantum systems.
\newblock \emph{Phys. Rev. A}, 82:022116, Aug 2010.

\bibitem[SPM13]{SPM13}
J.~Silman, S.~Pironio, and S.~Massar.
\newblock Device-independent randomness generation in the presence of weak
  cross-talk.
\newblock \emph{Phys. Rev. Lett.}, 110:100504, Mar 2013.

\bibitem[SW08]{ScholzW08}
V.~B. Scholz and R.~F. Werner.
\newblock Tsirelson's problem.
\newblock Technical report, {arXiv}:0812.4305v1 [math-ph], 2008.

\bibitem[UZZ{\etalchar{+}}13]{um2013experimental}
M.~Um, X.~Zhang, J.~Zhang, Y.~Wang, S.~Yangchao, D.-L. Deng, L.-M. Duan, and
  K.~Kim.
\newblock Experimental certification of random numbers via quantum
  contextuality.
\newblock \emph{Scientific reports}, 3, 2013.

\bibitem[Vid13]{Vidick13xor}
T.~Vidick.
\newblock Three-player entangled {XOR} games are {NP}-hard to approximate.
\newblock In \emph{Proc. 54th FOCS}. 2013.

\bibitem[Voi83]{Voi83}
D.~Voiculescu.
\newblock Asymptotically commuting finite rank unitary operators without
  commuting approximants.
\newblock \emph{Acta Sci. Math. (Szeged)}, 45:429--431, 1983.

\bibitem[VV12]{Vazirani2012}
U.~Vazirani and T.~Vidick.
\newblock {Certifiable quantum dice}.
\newblock \emph{Phil. Trans. R. Soc. A}, 2012.

\bibitem[YVB{\etalchar{+}}14]{YangVBSN14robust}
T.~H. Yang, T.~Vertesi, J.-D. Bancal, V.~Scarani, and M.~Navascues.
\newblock {Robust and Versatile Black-Box Certification of Quantum Devices}.
\newblock \emph{Phys. Rev. Lett.}, {113}({4}), {JUL 22} {2014}.

\end{thebibliography}

%\bibliography{../tex_headers/library}
\end{document}